\documentclass{llncs}
\usepackage{frameit,amsmath,color,latexsym,graphics,wrapfig,amssymb,balance}

\usepackage{times,epsfig}
\usepackage{graphicx}
\usepackage{ifthen}
\usepackage{times,epsfig,amsmath,clrscode3e,amssymb}
\usepackage{frameit,amsmath,color,latexsym,graphics,wrapfig,textcomp,dsfont,centernot,enumitem,mathtools,pifont,subcaption}
\usepackage{xcolor,colortbl}
\usepackage[utf8]{inputenc}

\usepackage{hyperref}
\usepackage{url}
\usepackage {tikz}
\usetikzlibrary{arrows,automata}
\usetikzlibrary{positioning}
\usepackage[linesnumbered,ruled,vlined,algo2e]{algorithm2e}
\usepackage{multicol}
\definecolor{Gray}{gray}{0.85}
\newcommand{\setvars}[1]{\ensuremath{\bar{#1}}}

\newcommand{\dfa}[5]{\ensuremath{\langle#1{,}#2{,} #3{,}#4{,}#5\rangle}}

\newcommand{\savespace}{\vspace{-2mm}}
\newcommand{\saveone}{\vspace{-1mm}}

\newcommand{\RA}{\ensuremath{\mathcal{R}}}

\mathchardef\mhyphen="2D

\newcommand{\seppredF}[2]{\ensuremath{#1(#2)}}
\newcommand{\seppred}[2]{\ensuremath{#1(#2)}}

\newcommand{\base}{\ensuremath{{\phi^b}}}

\newcommand{\report}[1]{ }

\newcommand{\acm}[1]{ }

\newcommand{\hide}[1]{}
\newcommand{\hideie}[1]{}

\newcommand{\pure}{\ensuremath{\pi}}

\newcommand{\astart}{\ensuremath{\Upsilon}}
\newcommand{\subterm}{\ensuremath{\Lambda}}

\newcommand{\myit}[1]{\textit{#1}}

\def\FV{\myit{FV}}

\def\fresh{\myit{fresh}}

\def\int{\code{int}}

\def\true{\code{true}\,}
\def\false{\code{false}\,}
\def\a{a}

 \newcommand{\code}[1]{{\small {\ensuremath{\tt #1}}}}

\newcommand{\btt}[1]{{\ensuremath{\tt #1}}}

\newcommand{\locnay}[1]{}
\newcommand{\wnnay}[1]{}
\newcommand{\longnay}[1]{}
\newcommand{\nodo}[1]{}

\newtheorem{thm}{Theorem}[section]
\numberwithin{thm}{section} 

\newtheorem{cor}[thm]{Corollary}
\newtheorem{defn}{Definition}

\def\S{\Sigma}

\def\SStore{\myit{SStacks}}
\def\IStore{\myit{ZStacks}}

\def\Stack{\myit{Stacks}}

\def\iffs{\small \btt{ iff~}}

\def\fresh{\myit{fresh}}

\newcommand{\atom}{\alpha}

\newcommand{\yields}{\leadsto}

\newcommand{\imply}{\ensuremath{\Rightarrow}}

\newcommand{\defsym}{\ensuremath{\overset{\text{\scriptsize{def}}}{=}}}

\newcommand{\sstack}{\ensuremath{\eta}}
\newcommand{\istack}{\ensuremath{\beta_{\eta}}}
\newcommand{\force}{\ensuremath{\models}}

\newcommand{\form}[1]{\ensuremath{#1}}

\newcommand{\Horn}{\RA\hide{\mathcal{C}}}
\newcommand{\PName}{\mathcal{P}}
\newcommand{\regex}{\mathcal{R}}

\newcommand{\sub}{\ensuremath{\sigma}}

 \newcommand{\sat}{{\sc{\code{SAT}}}}
 \newcommand{\unsat}{{\sc{\code{UNSAT}}}}
 \newcommand{\unknown}{{\sc{\code{UNKNOWN}}}}

\newcommand{\extractedtl}{{\sc{\code{{extract\_edtl}}}}}

\def\qed{\hfill\ensuremath{\square}}

\newcommand{\subst}[2]{\ensuremath{[#1 {/} #2]}}

\newcommand{\utree}[1]{\ensuremath{{\mathcal T_{#1}}}}
\newcommand{\postpro}{\code{postpro}}

\newcommand{\transform}{\code{reduce}}
\newcommand{\complete}{\code{complete}}

\newcommand{\lb}{\code{link\_back}}

\newcommand{\alphabet}{\Sigma}
\newcommand{\semp}{\epsilon}
\def\SVar{\myit{U}}
\def\SGVar{\myit{G}}
\def\SLVar{\myit{E}}
\def\IVar{\myit{I}}
\newcommand{\lang}{\sc{{L}}}
\newcommand{\strProb}{{\tt SAT{-}STR}}
\newcommand{\strel}{{\tt STR}}
\newcommand{\stredtsem}{{\tt STR^{dec}_{\edtl}}}
\newcommand{\stredt}{{\tt STR_{\edtl}}}
\newcommand{\strcfsem}{{\tt STR^{dec}_{\cfl}}}
\newcommand{\strcf}{{\tt STR_{\cfl}}}
\newcommand{\strflatsem}{{\tt STR^{dec}_{\fl}}}
\newcommand{\strflat}{{\tt STR_{\fl}}}
\newcommand{\allSat}{{\ensuremath{\omega}{-}\tt SAT}}

\newcommand{\se}{\ensuremath{\code{e}}}
\newcommand{\arith}{\ensuremath{\atom}}
\newcommand{\ariths}{\ensuremath{\arith}}
\newcommand{\seqs}{\ensuremath{\ses}}
\newcommand{\ses}{\ensuremath{\mathcal{E}}}

\newcommand{\Par}[1]{\ensuremath{{Par}(#1)}}
\newcommand{\widentree}[1]{\ensuremath{\code{widentree}(#1)}}

\newcommand{\extractpres}[1]{\ensuremath{{\sc extract\_pres}(#1)}}
\newcommand{\langg}[1]{\ensuremath{\mathcal{L}(#1)}}

\def\lang{\sc{L}}

\newcommand{\deci}{{\ensuremath{\tt Kepler_{22}}}}
\newcommand{\sleng}[1]{\ensuremath{|#1|}}

\newcommand{\tnode}{\btt{v}}

\newcommand{\cfgrammar}[4]{\ensuremath{{\langle}#1{,}~#2{,}~#3{,}~#4{\rangle}}}
\newcommand{\netgrammar}[4]{\ensuremath{{\langle}#1{,}~#2{,}~#3{,}~#4{\rangle}}}
\newcommand{\lgrammar}[4]{\ensuremath{{\langle}#1{,}~#2{,}~#3{,}~#4{\rangle}}}
\newcommand{\var}{V}
\newcommand{\production}{P}
\newcommand{\productions}{\mathcal{P}}

\newcommand{\red}{{\code{reduce}}}

\newcommand{\comp}{{\code{complete}}}
\newcommand{\matchse}{{\code{match}}}
\newcommand{\ctree}[3]{\ensuremath{{\mathcal C}(#1{\rightarrow}#2, #3)}}

\newcommand{\vtree}[4]{\ensuremath{{\mathcal C}(#1{\rightarrow}#2, #3)^{#4}}}

\newcommand{\backfun}{\ensuremath{{\mathcal C}}}
\newcommand{\cfgs}{\ensuremath{{\mathcal G}}}

\newcommand{\classlang}[1]{\ensuremath{{\mathcal L}(#1)}}
\newcommand{\classlangfin}[1]{\ensuremath{{\mathcal L}(#1)_{FIN}}}
\newcommand{\cflfs}{{\ensuremath \classlangfin{CF}}}
\newcommand{\cfls}{{\ensuremath \classlang{CF}}}
\newcommand{\cfl}{{\ensuremath {CFL}}}
\newcommand{\edtl}{{\ensuremath {EDT0L}}}
\newcommand{\cfg}{{{CFG}}}
\newcommand{\fl}{{\ensuremath {flat}}}
\newcommand{\edtlfs}{{\ensuremath \classlangfin{EDT0L}}}
\newcommand{\edtls}{{\ensuremath \classlang{EDT0L}}}
\newcommand{\etlfs}{{\ensuremath \classlangfin{ET0L}}}

\newcommand{\etl}{{\ensuremath {ET0L}}}
\newcommand{\transition}{{\ensuremath {\delta}}}

\newcommand{\conferencepaper}{1} 
\newcommand{\rep}[1]{\ifthenelse{\conferencepaper = 0}{#1}{}}
\newcommand{\repconf}[2]{\ifthenelse{\conferencepaper = 0}{#1}{#2}}
\begin{document}

\title{Decidable Logics Combining Word Equations,
 Regular Expressions and Length Constraints} 

\author{
Quang Loc Le
 }
  \institute{
 Teesside University, United Kingdom
  }


\maketitle

\savespace
\begin{abstract}

In this work,
we consider the satisfiability problem in
a logic that combines
word equations over string variables
 denoting words of unbounded lengths,
regular languages to
which words belong and
 Presburger constraints on the length of words.
We present a novel decision procedure over
 two decidable fragments that include quadratic word equations
  (i.e., each string variable occurs at most twice).
The proposed procedure
reduces
the problem to solving the satisfiability in the
 Presburger arithmetic.
The procedure combines two main components:
(i) an algorithm to derive 
a complete set
 of
all solutions of conjunctions of word equations
and regular expressions;
and (ii) two methods to precisely compute relational constraints
over string lengths
implied by the set of all solutions.
We have implemented a prototype tool
and evaluated it over a set of satisfiability problems in the logic.
The experimental results show that the tool
is effective and efficient.

\end{abstract}

\section{Introduction}\label{sec.intro}
The problem of solving
word algebras has been studied since the early stage
in mathematics and computer science \cite{Diekert2015}.
Solving
word equation (which includes
concatenation operation, equalities and
 inequalities on string variables)
was an intriguing problem and initially investigated due to its ties to Hilbert’s 10th problem.
The major result was obtained in 1977 by Makanin \cite{Makanin:math:1977}
who showed that the satisfiability of
word equations with constants is,
 indeed, decidable.
In recent years, due to considerable
number of security threats over the Internet,
there has been much renewed interest in the
satisfiability problem
involving the development of
 formal reasoning systems to either verify safety properties
or to detect vulnerability for web and database applications.
These applications often require
a reasoning about string theories that combines
word equations,
regular languages 
 and
constraints on the length of words. 

\hide{To support those reasoning systems above, 
several (either theoretical
or practical) decision procedures have been developed
that check
 satisfiability 
 of constraint languages over the string theories
(a.k.a. string solvers).
There have been a few decidability results in string theories.
For the satisfiability problem of  word equations,
besides Makanin's algorithm, two decision procedures, based on compression techniques,
\cite{Plandowski:ICALP:1998,Plandowski:STOC:1999,Jez:JACM:2016} have been
proposed recently.
Decidability of word equations and regular memberships
was shown in
\cite{Schulz:1990:MAW,Hooimeijer:PLDI:2009,Liang:FroCos:2015}
and implemented
in a Satisfiability Modulo Theories (SMT) framework \cite{Liang:FroCos:2015}.}

Providing a decision procedure
for the satisfiability problem on a string logic
including word equations and length constraints
has been difficult to achieve.
One main challenge is how to support an inductive
 reasoning about the combination of
 unbounded strings and the {\em infinite} integer domain.
Indeed, the satisfiability of word equations combined with length
constraints of the form \form{\sleng{x}{=}\sleng{y}} is open
\cite{Buchi1990,Ganesh:HVF:2012} (where \form{\sleng{x}} denotes
the length of the string variable \form{x}).
So far, very few decidability results in this logic 
are known; the most expressive result is 
restricted
within the straight-line fragment (SL) which is based on {\em acyclic} word equations
\cite{Ganesh:HVF:2012,Abdulla:CAV:2014,Lin:POPL:2016,Tao:POPL:2018,Lukas:POPL:2018}.
This SL fragment excludes 
constraints combining
 {\em quadratic} word equations, the equations in which
 each string variable occurs at most twice.
For instance, the following constraint is beyond
the SL fragment:
$\se_c{\equiv}x{\cdot}a {\cdot} a {\cdot} y ~{=}~ y{\cdot}b{\cdot}a{\cdot}x$ where \form{x} and \form{y}
are string variables, 
 \form{a} and \form{b} are letters,
and \form{\cdot} is the string concatenation operation.
Hence, one research goal
 is to identify decidable logics combining quadratic word equations (and beyond), 
based on which we can develop an efficient decision procedure.

\hide{
While 
the state-of-the-art
solvers, e.g. Z3str2 \cite{Zheng:FSE:2013,Zheng:CAV:2015},
  Norn \cite{Abdulla:CAV:2014,Abdulla:CAV:2015},
CVC4 \cite{Liang:FMSD:2016}, S3P \cite{Trinh:CAV:2016} and TRAU \cite{Parosh:PLDI:2017}, support very well the verification of
 software in computer security communities, 
they may not terminate
over those satisfiability problems including {\em cyclic}
 word equations e.g., quadratic word equations in which
 each string variable occurs at most twice.
In particular, the algorithms presented in \cite{Zheng:FSE:2013,Liang:CAV:2014,Abdulla:CAV:2014,Abdulla:CAV:2015,Trinh:CAV:2016}
 are unable to decide the satisfiability for
 the following quadratic word equation which
 has two occurrences of each string variable \form{x} and \form{y}:
$\se_c{\equiv}x{\cdot}a {\cdot} a {\cdot} y ~{=}~ y{\cdot}b{\cdot}a{\cdot}x$
where \form{\cdot} is the string concatenation operation,
 \form{a} and \form{b} are letters.}

\hide{
For efficiency, heuristics were recently introduced in \cite{Zheng:CAV:2015}, 
\cite{Liang:FMSD:2016},
 \cite{Trinh:CAV:2016}
and \cite{Parosh:PLDI:2017} to avoid such non-termination
for some special scenarios.
However, these approaches are far away from providing complete string solvers
for a sufficiently generic string theory such as the one including
 quadratic word equations.}

\hide{One way to support a decision procedure
for string constraints combining word equations 
and length functions is to precisely compute
the length constraints implied by
 all solutions of each word equation. 
However, this task is nontrivial
as it requires the procedure to derive
the set of all solutions
as well as precisely extract the length constraints. 
While the set of all solutions could be derived by
 the existing works like \cite{Jaffar:JACM:1990,Plandowski:STOC:2006,Jez:JACM:2016,Ciobanu:ICALP:2015},
computing the length constraints for 
the solution sets derived by these works is not straightforward.
For instances, while the algorithm in \cite{Jaffar:JACM:1990}
does not terminate when the set is infinite,
 the length constraints
 derived by \cite{Plandowski:STOC:2006,Jez:JACM:2016}
may not in a finite form.
In particular, although Plandowski's work \cite{Plandowski:STOC:2006}
can derive a finite representation of all solutions
of a word equation, the length constraints implied by this representation
 are not always represented with finitely many
 equations in numeric solvable form.
}

There have been efforts to deal with the cyclic string constraints
 in Z3str2 \cite{Zheng:CAV:2015,Zheng:FMSD:2017}, CVC4 \cite{Liang:FMSD:2016} and S3P \cite{Trinh:CAV:2016}.
While Z3str2
presented a mechanism to detect overlapping variables to avoid
non-termination,
 CVC4 proposed {\em refutation complete} procedure
to generate a refutation for any unsatisfiable input problem
and
 S3P \cite{Trinh:CAV:2016} provided a method
to identify and prune non-progressing
scenarios.
However, 
none is both complete and terminating over
quadratic word equations.
For instance, Z3str2, CVC4
and S3P (and all the
state-of-the-art string solving techniques 
\cite{Abdulla:CAV:2014,Abdulla:CAV:2015,Parosh:PLDI:2017,Berzish:FMCAD:2017,Tao:POPL:2018,Lukas:POPL:2018}) is not able to decide the satisfiability
of the word equation
 \form{\se_c} above.

In this work,
we propose a novel cyclic proof system within
a satisfiability procedure
for the string theory combining word equations, regular memberships
and Presburger constraints over the length functions.
Moreover, we identify decidable fragments with quadratic word equations
 (e.g., the constraint \form{\se_c} above)
where the proposed procedure is complete and terminating.
To the best of our knowledge,
our proposal is the first
decision procedure for string constraints beyond
the straight-line word equations.
Our proposal has two main components.
First,  we present
a novel algorithm
to construct a cyclic {\em reduction tree} which finitely represents
all solutions of a conjunction of word equations and regular membership predicates.
Secondly, we describe two procedures
to infer
the length constraints implied by the set of all solutions.

\hide{We remark that although
the procedure in \cite{Ciobanu:ICALP:2015}
is able to derive for each word equation
a description, a tree, of ETOL language,
 it is unclear whether this description
is a language of
finite index.
Furthermore, node of the tree
  derived by \cite{Ciobanu:ICALP:2015}
is extended equation which is an element in
a free partially commutative monoid rather
than simply a word equation.
Hence, inferring the length constraints
implied by all solutions of a word equation using the procedure
 in
\cite{Ciobanu:ICALP:2015}
is not straightforward.
}

\savespace
\paragraph{Contributions.} We make the following
 technical contributions.
 \begin{itemize}[noitemsep,topsep=1pt]
 \item 
We develop a novel
 algorithm, called {\allSat},
 to derive 
 a finite representation for
all solutions
of a conjunction of word equations and regular expressions.
\item We
 present a 
 decision procedure, called {\deci}, with 
two decidable fragments
and provide a complexity analysis of our approach.
This is the first decidable result for the string theory
 combining {\em quadratic} word equations with length constraints.
\item We have implemented a prototype solver
and evaluated it over a set of hand-drafted benchmarks in
the decidable fragments.
The experimental results show that our proposal is
both effective
 and efficient
in solving string constraints with quadratic word equations and length constraints.
\end{itemize}
\savespace
\paragraph{Organization.}
The rest of the paper is organized as follows.
Sect \ref{sec.prelim} presents relevant definitions. Sect \ref{sec.motivate}
 shows an overview of our approach through an example. We show how to
compute a cyclic reduction tree
to finitely represent
all solutions of a conjunction of word equations and regular memberships
 in Sect \ref{sec.all.sol}.
Sect \ref{sec.deci} presents the proposed decision procedure.
Sect \ref{sec.quad} and Sect \ref{sec.flat} describe
 the two decidable fragments. 
Sect \ref{sec.impl} presents
an implementation and evaluation.
 Sect \ref{sec.related} reviews related work and 
 concludes. For the space reason,
all missing proofs are presented in Appendix.

\section{Preliminaries}\label{sec.prelim}



Concrete string models assume a finite alphabet \form{\alphabet}
whose elements are called {\em letters},
set of finite words over \form{\alphabet^*} including \form{\semp}
 - the empty word, and
a set of integer numbers \form{\mathbb{Z}}.
 We work with a set \form{\SVar}
of string variables denoting words in \form{\alphabet^*},
and a set \form{\IVar} of arithmetical variables.
We use
\form{\sleng{w}} to denote the length of \form{w{\in} \alphabet^*}
and
\form{\setvars{v}} a sequence of variables.
A language \form{\lang} over the alphabet \form{\alphabet}
is a set \form{\lang {\subseteq} \alphabet^*}.
A language \form{L} is a set of words generated
by a grammar system.
We use \form{\classlang{L}} to denote the class of all 
languages \form{L}.

 \begin{figure}[tb]
 \begin{center} 
\savespace\[\savespace
\begin{array}{c}
\begin{array}{lcll}
\text{disj formula} & \pure &::=& \phi  ~|~ \pure_1 ~{\vee}~ \pure_2  \qquad
\text{formula} \quad \phi ~~::= ~~ \se \mid \arith 
                \mid s{\in} \regex
                 ~|~ \neg\phi_1 \mid 
\phi_1 ~{\wedge}~ \phi_2  \\
\text{(dis)equality} & \se &::=&  s_1 {=} s_2  
\qquad \qquad \text{term} \quad s ~~::=~~  \semp ~|~ c \mid x 
 \mid s_1 \cdot s_2\\
\text{regex} & \regex &::=&  \emptyset \mid \semp ~|~ c \mid w
 \mid \regex_1 \cdot \regex_2 \mid \regex_1 + \regex_2 \mid 
   \regex_1 \cap \regex_2 \mid \regex_1^C \mid \regex_1^*\\
 \text{Arithmetic}
 &\arith &::=& \a_1 = \a_2 \mid \a_1 > \a_2 \mid \arith_1 \wedge \arith_2
  \mid \arith_1 \vee \arith_2 \mid \exists v . \arith_1  \mid \seppredF{\code{P}}{\setvars{v}} \\
 & a &::=& \!\!
    \begin{array}[t]{l}
      0 \mid 1 \mid v \mid \form{\sleng{u}} \mid i \times \a_1  \mid {-}\a_1 \mid 
       \a_1 + \a_2  
    \end{array}\\
 \end{array} \\
\end{array}
\]
 \saveone \savespace
\caption{Syntax}\label{prm.spec.fig}
\end{center} \savespace \savespace \saveone \savespace \savespace \savespace
\end{figure}

\noindent{\em \bf Syntax} The syntax of quantifier-free string formulas, called {\strel},
 is presented in Fig. \ref{prm.spec.fig}.
 \form{\pure} is a disjunction formula where each disjunct 
\form{\phi}
is a conjunction of word equations \form{\se}, 
regular memberships \form{s{\in} \regex}
and  arithmetic constraints \form{\arith}.
Especially, \form{\arith} may contain predicates
\form{\seppredF{\code{p}}{\setvars{v}}} whose definitions are inductively defined.
We use either \form{s} or \form{tr} to denote a string term.
 We often
write \form{s_1s_2} to denote \form{s_1\cdot s_2}
if it is not ambiguous.
Regular expression \form{\regex} over \form{\alphabet}
is built over \form{c \in \alphabet}, \form{w \in \alphabet^*}, \form{\semp}, and closing
under union \form{{+}}, intersection \form{\cap},
complement \form{C}, concatenation \form{\cdot},
and the Kleene star operator \form{*}.
Regular expressions \form{\regex} does not contain any string variables.

 We use 
 \form{\seqs} to denote a conjunction (a.k.a system) of word equations.
 \form{\pure\subst{t1}{t2}} denotes a
substitution of  all occurrences
of \form{t_2} in \form{\pure} to \form{t_1}.
We use function \form{\FV(\pure)} to return all free variables
of \form{\pure}.
%
We inductively define {length function of a string term \form{s},
denoted as \form{\sleng{s}}, 
as:}
\form{
\sleng{\semp}=0}, \form{\sleng{c}= 1,}~
and
\form{\sleng{s_1\cdot s_2}=\sleng{s_1}+ \sleng{s_2}
}.
%
\hide{In the following, we define the function \code{unfold\_str}
used to derive an equivalent disjunction
of a formula through unfolding one its string predicate occurrence .
\[
\begin{array}{c}
\pure_{base} {\equiv} \pure \subst{\semp}{\seppred{\code{STR}}{u{,}l{,}r}} {\wedge}l{=}r{\wedge}u{=}\semp  \qquad \fresh~ u_1{,}l_1\\
\pure_{ind} {\equiv} \pure \subst{c{\cdot}\seppred{\code{STR}}{u_1{,}l_1{,}r}}{\seppred{\code{STR}}{u{,}l{,}r}} {\wedge}l_1{=}l{+}1{\wedge}l{<}r {\wedge}u{=}c{\cdot}u_1
\\
\hline
\code{unfold\_str}( \seppred{\code{STR}}{u{,}l{,}r} {\wedge} \pure) ~{\yields}~
 \pure_{base} ~{\vee}~\pure_{ind}
\end{array}
\]
where \form{\pure\subst{t1}{t2}} substitutes all occurrences
of \form{t_2} in \form{\pure} to \form{t_1},
and \form{\fresh ~ \setvars{v}} returns a set of fresh variables
\form{\setvars{v}}.
 The procedure \code{unfold\_str} outputs a set of two disjuncts which are combined from branches of the predicate
\code{STR} with the remainder \form{\pure}.
The constraint \form{u{=}\semp} in base branch
and \form{u{=}c{\cdot}u_1} in inductive branch
are denoted as {\em subterm} constraints. They are
used for constructing a model to witness satisfiability.}
%
Notational length 
 of the word equation \form{\se}, denoted by \form{\se(N)},
is the number of its symbols.

A word equation is called {\em acyclic} if
each variable occurs at most once.
A word equation is called {\em quadratic}
if each variable occurs at most twice.
Similarly, a system of word equations is called
quadratic if each variable occurs at most twice.


A word equation system is said to be straight-line  \cite{Ganesh:HVF:2012,Abdulla:CAV:2014,Lin:POPL:2016}
if it can be rewritten (by reordering the conjuncts)
as the form \form{\bigwedge_{i{=}1}^n x_i = s_1}
such that:
(i)
\form{x_1},...,\form{x_n} are different variables;
and (ii) \form{\FV(s_i) \subseteq \{x_1,x_2,..,x_{i{-}1} \}}.
 A formula \form{\pure \equiv \se_1 {\wedge}\se_2
{\wedge} ...{\wedge}\se_n \wedge \astart}
is called in straight-line fragment (SL)  if
 \form{\se_1 {\wedge}\se_2
{\wedge} ...{\wedge}\se_n} is straight-line and
the regular expression \form{\astart} is of the conjunction
of regular memberships
\form{x_j {\in} \regex_j} where \form{x_j {\in} \{x_1,...,x_n\} }.





\noindent{\em \bf Semantics}
Every regular expression \form{\regex}
is evaluated to the language \form{\classlang{\regex}}.
We define:
\savespace\[\savespace
\begin{array}{lcllcl}
{\SStore}  & {\defsym} &  {(\SVar {\cup} \alphabet)} {\rightarrow} \alphabet^*  & \qquad
{\IStore} & {\defsym} &  {\IVar} ~{\rightarrow}~ \mathbb{Z}
\end{array}.
\]
The semantics is given by a satisfaction relation:
\form{\sstack{,}\istack {\force} \pure}
that forces the interpretation on both string \form{\sstack} and
 arithmetic \form{\istack} to satisfy the constraint
 \form{\pure}
 where \form{\sstack \in {\SStore} }, \form{\istack {\in} {\IStore}},
and \form{\pure} is a formula.
We remark that \form{\forall \sstack \in {\SStore}}:
\form{\sstack(c){=}c}
for all \form{c \in \alphabet}
and \form{\sstack(t_1 t_2){=}\sstack(t_1) \sstack(t_2)}.
 The semantics of our language is formalized in App. \ref{app.sem}.
 Inductive predicate is interpreted as
 a least fixed-point of values \cite{Makoto:APLAS:2016}.
%
If \form{\sstack{,}\istack \force \pure},
we use the pair \form{\langle\sstack{,}\istack\rangle} to denote
a solution of the formula \form{\pure}.
Let
 \form{\se{\equiv}x_1{\cdot}...{\cdot}x_l{=}x_{l{+}1}{\cdot}...{\cdot}x_n}
be  a word equation.
If \form{\se} is satisfied with the solution
 \form{\langle\sstack{,}\istack\rangle}, we also refer
\form{\sstack(x_1){\cdot}...{\cdot}\sstack(x_l)}
as a solution word 
 of \form{\se}.
A solution word is minimal if the length of the
 solution word ($\sleng{\sstack(x_1)} + ... + \sleng{\sstack(x_l)}$)
  is minimal.
 \form{\se_1} is referred as a suffix
of  \form{\se_2} if
they are satisfied
and the solution word of \form{\se_1}
is a suffix
of the solution word of \form{\se_2}.


\noindent{\em \bf Formal Language} A deterministic finite automaton (DFA) A
is a tuple: \form{A{=}\dfa{Q}{\alphabet}{\transition}{q_o}{Q_F}},
where \form{Q} is a finite set of states,
\form{\transition \subseteq Q {\times} (\alphabet {\cup} \{\semp\}) {\times} Q} is a finite set of transitions, \form{q_0 {\in} Q} is the initial state
and \form{Q_F {\subseteq} Q} is a set of accepting states.
We use \form{\classlang{A}} to denote the (regular) language generated
by a DFA \form{A}.
It is known that
the languages generated by regular expressions are also in
the class of regular languages \cite{Hopcroft:2006:IAT}.

 A context-free grammar (CFG) G is defined by the quadruple:
\form{G{=}\cfgrammar{\var}{\alphabet}{\production}{S}} where
\form{\var} is a finite nonempty set of nonterminals,
\form{\alphabet} is a finite set of terminals and
 disjoint from \form{\var}, and
\form{\production {\subseteq} \var {\times} (\var {\cup} \alphabet)^*}
 is a finite relation.
For any strings \form{u,v {\in} (\var {\cup} \Sigma )^{*}},
\form{v} is a result of applying the rule \form{(\alpha ,\beta )} to \form{u}
\form{u\Rightarrow_G v\,}
  if \form{\exists (\alpha ,\beta )\in \production}
 \form{u_{1},u_{2}\in (V\cup \Sigma )^{*}}
 such that \form{ u\,=u_{1}\alpha u_{2}} and
 \form{v\,=u_{1}\beta u_{2}}.
\form{\langg{G} {=} \{w \in \alphabet^* \mid S \Rightarrow^*_G w \}} to denote a language produced
by the CFG \form{G}.
%
Given a CFG \form{G{=}\cfgrammar{\var}{\alphabet}{\production}{S}},
we use \form{G_X} (where \form{X \in \var}) to denote a sub-language of \form{\langg{G}}, defined
by \form{\langg{G_X} {=} \{w\in \alphabet^* ~|~ X \Rightarrow^*_G w \}}.

\noindent{\em \bf Normal Form}
\form{\pure{\equiv}\ses{\wedge} \astart{\wedge}\ariths}
 is called in the normal form if it is of the form:
\form{\ses} is a system of word equations,
\form{\astart} is a conjunction of regular memberships 
(e.g., \form{X {\in} \regex)} 
and \form{\ariths} is a Presburger formula.
The normalization procedure is left in
App. \ref{app.norm.form}.

\noindent{\em \bf Problem Definition} Throughout this work, we consider the following problem. 

\[
\begin{array}{|ll|}
\hline
 \text{PROBLEM:} & {\strProb}.\\
 \text{INPUT:} & \text{A string constraint }\form{\pure}  \text{ in normal form over } \form{\alphabet}. \\
 \text{QUESTION:} & \text{Is } \form{\pure} \text{ satisfiable? } \\
\hline
\end{array}
\]

Authors in  \cite{Ganesh:HVF:2012,Abdulla:CAV:2014,Lin:POPL:2016}
show that this problem in straight-line fragment is decidable.

\section{Overview and Illustration}\label{sec.motivate}
Overall of our idea is an algorithm to
reduce an input constraint to a 
 set of solvable constraints.
In this section, we first define the reduction tree
 (subsection \ref{mov.derv.tree}).
After that, we illustrate the proposed decision procedure through
an example (subsection \ref{sec.mov.example}).

\hide{Given a string formula (of our
decidable fragments) which is a conjunction
of a word equation system \form{W},
 regular memberships \form{Re} and Presburger arithmetic \form{A}
the proposed procedure works as follows.
\begin{itemize}
\item First, it constructs a reduction tree
as a finite representation of the {\em complete} set of all solutions of the input word equation system \form{W}.
\item Secondly,
our procedure extends the tree
to capture all solutions of the conjunction
 \form{W \wedge Re}.
\item
Thirdly, it equivalently transforms
the final tree into a {\em finite-index} EDT0L system
\cite{L:Book}. The acronym EDT0L refers to
{\bf E}xtended, {\bf D}eterministic,
 {\bf T}able, {\bf 0} interaction,
and {\bf L}indenmayer.
Intuitively, grammar systems that are
{\em k-index} are restricted so that, for
every word generated by the grammar, there is some
successful derivation where
at most {\em k} variables (or nonterminals)
appear in every sentential form
of the derivation \cite{ROZENBERG:1978}.
A system is finite-index if it is \form{k}-index
for some \form{k}.
\item Fourthly, it
computes the length constraints
of every string variables
which are precisely implied by
all words generated from the finite-index EDT0L system.
These length constrains, say \form{A_w}, are computed
as an existentially quantified Presburger formula.
This computation is based on
 Parikh's Theorem \cite{Parikh:JACM:1966}, one of the most
celebrated theorem in automata theory.
%
\item Lastly, it conjoins the length constraints \form{A_w} 
with the Presbuger constraints \form{A} to obtain conjunction
\form{A_w \wedge A}.
The satisfiablity problem of the input formula \form{W \wedge Re \wedge A}
is {\em equi-satisfiable} with the Presbuger constraint \form{A_w \wedge A}.
As Presbuger arithmetic is decidable,
 the satisfiability problem for the aforementioned 
fragment is decidable.
\end{itemize}}
\savespace
\subsection{Cyclic Reduction Tree}\label{mov.derv.tree}
 Formally, a cyclic reduction tree \form{\utree{i}} is 
 a tuple \form{(V, E, \backfun)} where
 $V$ is a finite set of nodes where each node represents a conjunction
of word equations \form{\ses}.
 $E$ is a set of labeled and directed edges  \form{(\ses, \sub, \ses') \in E} where \form{\ses'} is a child of \form{\ses}.
This edge means
 we can reduce \form{\ses} to \form{\ses'} via the label \form{\sub}, a substitution,
s.t.: \form{\ses' ~{\equiv}~ \ses\sub}.
And $\backfun$ is a back-link (partial) function which captures virtual cycles in the tree.
 A cycle, e.g. \form{\ctree{\ses_c}{\ses_b}{\sub}},
 in $\backfun$ means
the leaf \form{\ses_b} is linked back to its ancestor \form{\ses_c} and
 \form{\ses_c ~{\equiv}~ \ses_b\sub}.
In this back-link, \form{\ses_b} is referred as a {\em bud}
and \form{\ses_c} is referred as a {\em companion}.
%
A path \form{(\tnode_s,\tnode_e)}
is a sequence of nodes and edges
connecting node \form{\tnode_s} with node
\form{\tnode_e}.
A leaf node is either unsatisfiable,
or satisfiable or linked back to an interior node, or
not-yet-reduced.
If a leaf node is not-yet-reduced, it is marked as open.
Otherwise, it is marked as closed.
A trace of a tree
 is a sequence of edge labels of a path in the tree.
We refer 
 a trace as solution trace if it corresponds to a
path \form{(\tnode_s,\tnode_e)} where \form{\tnode_s}
is the root
and \form{\tnode_e} is a satisfiable leaf.
This trace represents a (infinite) family solutions of the equation at the root.
\savespace
\subsection{Illustrative Example}\label{sec.mov.example}
\begin{figure}[tb]
\begin{center}
\begin{scriptsize}
 \savespace
\begin{tikzpicture}[node distance=18mm,level 1/.style={sibling distance=22mm},
      level 2/.style={sibling distance=25mm},
                        level distance=22pt, draw]
  \tikzstyle{every state}=[draw,text=black]

\node (A)                    {\scriptsize \textcolor{blue}{$\ses_0{\equiv}abx {=} xba{\wedge}ay{=}ya^\bigstar$}};
  \node         (B) [below =4mm of A] {\scriptsize \underline{$\ses_{11}{\equiv}ab{=}ba{\wedge}ay{=}ya$}};
  \node         (C) [below right=3mm and 4mm of A] {\scriptsize{$\ses_{12}{\equiv}bax_1 {=} x_1ba{\wedge}ay{=}ya$}};
  \node         (D) [below left=4mm and 2mm of C] {\scriptsize \textcolor{blue}{${\ses_{21}{\equiv}
      ay{=}ya^\heartsuit}$}};
\node         (D1) [below left=3mm and 2mm of D] {\scriptsize${\ses_{31}{\equiv}
      \semp{=}\semp}$};
\node         (D2) [below right=3mm and 2mm of D] {\scriptsize \textcolor{blue}{${\ses_{32}{\equiv}
      ay_1{=}y_1a^\heartsuit}$}};
  \node         (E) [below right=3mm and 1mm of C] {\scriptsize\textcolor{blue}{$\ses_{22}{\equiv}abx_2 {=} x_2ba{\wedge}ay{=}ya^\bigstar$}};

  \path (A) edge[->]     node [pos=0.3,left] {\form{\small [\semp/x]}} (B)
            edge[->]     node [pos=0.3,right] { \form{\small [ax_1/x]}} (C)
        (C) edge [->]    node [pos=0.3,left] { \form{\small [\semp/x_1]}} (D)
            edge [->]    node [pos=0.3,right] { \form{\small[bx_2/x_1]}} (E)
        (D) edge [->]    node [pos=0.3,left] {\form{\small[\semp/y]}} (D1)
            edge [->]    node [pos=0.1,below] {\form{\small[ay_1/y]}} (D2)
        (D2) edge [->,bend right=30,dotted]  node [pos=0.3,above] {  \form{\small[y/y_1]}} (D)
        (E) edge [->,bend right=30,dotted]  node [pos=0.3,above] {\form{\small[x/x_2]}} (A);
\end{tikzpicture}
 \saveone  
\end{scriptsize}
\caption{Reduction Tree $\utree{3}$.}\label{fig.unfolding.tree}
\end{center}
  \savespace \savespace \savespace
 \savespace \savespace 
\end{figure}
We consider the following constraint:
\savespace\[\saveone
\begin{array}{l}
 \form{\pure~{\equiv}~abx{=}xba {\wedge}
      ay{=}ya  
~{\wedge}~
  (\exists k . \sleng{x} {=} 4k{+}3) {\wedge}
  \sleng{x} {=} 2\sleng{y}}
 \end{array}
\]
where \form{x}, \form{y}  are string variables and \form{a}, \form{b} are letters. This constraint is beyond the straight-line fragment
\cite{Ganesh:HVF:2012,Abdulla:CAV:2014,Lin:POPL:2016,Tao:POPL:2018,Lukas:POPL:2018}.
Moreover, as the length constraint \form{\sleng{x} {=} 2\sleng{y}}
is not regular-based, the automata-based translation proposed in \cite{Tao:POPL:2018} cannot be applied.

The proposed solver
{\deci} could solve the constraint \form{\pure} above
 through the following
three steps.
First, it invokes procedure {\allSat} to construct
a cyclic reduction tree  to
capture
  all solutions
 of the word equations
 \form{\ses_0{\equiv}abx {=} xba{\wedge}
      ay{=}ya}.
Next, it
 infers
  a precise constraint 
 \form{\arith_{xy}} implied by string lengths of all solutions.
Lastly, it solves the conjunction:
\form{{\arith}_{xy} {\wedge} \arith} 
where \form{\arith} is the arithmetic constraint in the input
\form{\pure}. 

\paragraph{The representation of all solutions}
%
{\allSat} derives the reduction tree $\utree{3}$ $(V, E, C)$,
shown in Figure \ref{fig.unfolding.tree},
as the finite presentation of all solutions
 for \form{\ses_0}.
In particular, the root of the tree is
 \form{\ses_0}.
\form{\ses_0} has two children \form{\ses_{11}}
and $\ses_{12}$, which are obtained
 by reducing  \form{x}
  into two {\em complete} cases:
\form{x{=}\semp} and \form{x{=}ax_1} where \form{x_1} is fresh.
 Note that $\ses_{12}$ is obtained 
by first 
applying the substitution:
$\ses'_{12}{\equiv}\ses_0[ax_1/x]{\equiv}\textcolor{blue}{a}bax_1 {=} \textcolor{blue}{a}x_1ba{\wedge}ay{=}ya$
prior to
subtracting
the letter \form{a} at the heads of the two sides of 
the first word equation.
Next, while \form{\ses_{11}} is classified
as unsatisfiable, (underlined) and marked closed,
 \form{\ses_{12}} is further reduced 
into
 two children,
\form{\ses_{21}} and \form{\ses_{22}}. They are obtained by
 reducing \form{x_1} at the head of
the right-hand side (RHS) of \form{\ses_{12}} into
 two complete cases: \form{x_1{=}\semp}
to generate \form{\ses_{21}'{\equiv}\ses_{12}'[\semp/x_1]{\equiv} \textcolor{blue}{ab}{=}\textcolor{blue}{ab}{\wedge}ay{=}ya}
 and \form{x_1{=}bx_2} (where \form{x_2} is a fresh variable)
to generate \form{\ses_{22}'{\equiv}\se_{12}'[bx_2/x_1]{\equiv}\textcolor{blue}{b}abx_2 {=} \textcolor{blue}{b}x_2ba}.
Next,
 \form{\ses_{21}'}
is further reduced into \form{\ses_{21}}
by matching \form{a}, \form{b} letters;
and \form{\ses_{22}'} is further reduced into \form{\ses_{22}}
by matching \form{b} letters at the heads of its two sides.
Lastly, 
\form{\ses_{22}} is linked back to \form{\ses_0}
to form the back-link \form{\ctree{\ses_0}{\ses_{22}}{[x/x_2]}}.
Similarly, \form{\ses_{21}}
is
 reduced until all leaf nodes are marked closed.

\hide{
Given a solution trace \form{\sub},
a solution of the root can be obtained through
the following
 constraint: \form{\bigwedge \{tr{=}tr' \mid [tr'/tr] \in \sub\}}.
 A cycle in the tree thus represents a set of infinite solutions,
 since we can construct infinitely many paths by iterating through the cycle an unbounded number of times.
 For instance, in Fig. \ref{fig.unfolding.tree}, we can obtain a solution for \form{\ses_0} following the 
path \form{(\ses_0,\ses_{31})}.
The trace of this path is: \form{\sub_{31}{=}[ax_1/x,\semp/x_1,\semp/y]}.
The solution \form{x{=}a} and \form{y{=}\semp} can be obtained
through the constraint:
 \form{x{=}ax_1 {\wedge} x_1{=}\semp{\wedge}y{=}\semp}.
Moreover,
we can obtain another solution following the cycle
 from \form{\ses_{0}} to \form{\ses_{22}} and back to \form{\ses_{0}}
 for a different number of times and
 then following the edge from \form{\ses_{0}} to \form{\ses_{31}}. 
As the trace and the constraint for each iteration of the cycle
is \form{\sub^{cyc}_{22}{=}\{ax_1/x,bx_2/x_1,x_2/x \}} 
and \form{x{=}ax_1{\wedge}x_1{=}bx_2{\wedge}x_2{=}x}, respectively, the set of
infinite solutions is of
the form: \form{x{=}a} and \form{y{=}\semp},
\form{x{=}{ab}\cdot a} and \form{y{=}\semp}, \form{x{=}{ab}\cdot ab \cdot a}
and \form{y{=}\semp},....}

A path \form{(\tnode_s,\tnode_e)} with trace \form{\sub}
represents for \form{\tnode_e{\equiv}\tnode_s\sub}.
If \form{\tnode_e} is satisfiable, then \form{\sub} represents for
 a family of solutions (or valid assignments).
For instance, in Fig. \ref{fig.unfolding.tree},
the 
path \form{(\ses_0,\ses_{31})}
has the trace \form{\sub_{31}{=}[ax_1/x,\semp/x_1,\semp/y]}.
As \form{\ses_{31}} is satisfiable, we can derive
a solution of \form{\ses_0} based on \form{\sub_{31}} as:
\form{x{=}a} and \form{y{=}\semp}.
Moreover, trace solution that is involved in
cycles represents a set of infinite solutions,
 since we can construct infinitely many solution traces by iterating through the cycles an unbounded number of times.
For example,
 all solution traces \form{\sub_{ij}}
 obtained from the path \form{(\ses_0,\ses_{31})} above
is as:
\savespace\[\savespace
\form{\sub_{ij}} {\equiv}~ \form{[ax_1/x] \circ [bx_2/x_1, x/x_2, ax_1/x ]^{i} \circ [ay_1/y, y_1/y]^{j}  \circ [\semp/x_1  \circ \semp/y]}
\]
where \form{\circ} is the substitution composition operation, \form{\sub^k} means
\form{\sub} is repeatedly composed zero, one or more times,
and \form{i{\geq}0}, \form{j{\geq}0}.

\hide{\noindent{\it {\edtl} Transformation}
A reduction tree derived by {\allSat} represents for
a finite-index {\edtl}
system.
For example, the tree $\utree{3}$ above represents
for the language:
 \form{G{=}\lgrammar{\{S, x, x_1,x_2\}}{\alphabet}{\{P_1,P_2\}}{S}}
where
\form{\production_1=\{(S, abx), (x,ax_1), (x1,\semp) \}}
and
\form{\production_2=\{(S, abx), (x,ax_1), (x1,bx_2), (x2,x) \}}.}

\hide{One fundamental task in this work is
to infer a precise arithmetical constraints
over length functions of string variables
corresponding to the infinite sets of solutions as above.
Inspired by Parikh's Theorem \cite{Parikh:JACM:1966}
which states that
 Parikh images of
context-free languages
precisely coincide with semilinear sets
(which in turn can be computed in Presburger arithmetic),
we first transform the cyclic reduction tree
into a (union) set of CFG and then compute the Parikh images for these CFG.
In particular, our algorithm derives the CFG through
solution paths which end with satisfiable leaf nodes.
For such a path with trace \form{\sub_i},
it generates a CFG G \form{\cfgrammar{\var_i}{\alphabet}{\production_i}{S_i}}
(where \form{\classlang{G}}
is the context-free language produced by \form{G})
as follows. First, it collects all traces,
corresponding to the cycles intersecting with the path.
Assume that there are \form{k} such cycles, \form{k{\geq}0},
and \form{\sub^{cyc}_{i}} is the composition of their traces
 (\form{\sub^{cyc}_{i}{=}\emptyset} if \form{k{=}0}).
Secondly, our algorithm introduces a set of new variables \form{\var_{c_i}}
and substitutions \form{\sub_{c_i}} to connect these cycles;
and  \form{\var_{c_i}} also contains
a new variable to connect the last cycle with the path \form{\sub_i}.
Finally,  \form{G} is generates as follows.
\form{\var_i} is the set of all variables
of the path, the \form{k} cycles, and the new variables  \form{\var_{c_i}}.
\form{\production_i{=} \bigcup \{tr \rightarrow tr' \mid  [tr'/tr] \in
 (\sub_i \cup \sub_{c_i} \cup \sub^{cyc}_{i})\}}.
For instance, our algorithm generates
for the path \form{(\se_0,\se_{21})} with the trace \form{\sub_{21}}
a CFG \form{G_1} \form{\cfgrammar{\var_1}{\alphabet}{\production_1}{S_1}} where \form{S_1} is a fresh variable to represent
for the start symbol),
\form{\var_1{=}\{S_1,x,x_1,x_2,x_3\}} and
the production function \form{P_1} is defined as:
\[
\begin{array}{c}
S_1 \quad \rightarrow \quad  abx \\
\begin{array}{ll}
\begin{array}{lcl}
x & \rightarrow & ax_1 \\
x_1 & \rightarrow & bx_2\\
x_2 & \rightarrow & x \\
\end{array}
&\qquad
\begin{array}{lcl}
x & \rightarrow &  x_3 \\
x_3 & \rightarrow & ax_1 \\
x_1 & \rightarrow & \semp \\
\end{array}
\end{array}
\end{array}
\]
The set of there production rules on the left-hand side represents
 the cycle
\form{\ctree{\se_0}{\se_{22}}{[x/x_2]}} with the trace \form{\sub^{cyc}_{22}}.
The set of three production rules on the right-hand side
represents the solution trace \form{(\se_0,\se_{21})}.
We notice that \form{x_3} (resp., \form{[x_3/x]})
is the new variable (resp., substitution)
 introduced to connect
the cycle with the path \form{(\se_0,\se_{21})}.
As can be seen, after connected, labels in the path \form{(\se_0,\se_{21})}
have also been renamed. We present this transformation in details in Section
\ref{algo.extractcfg}.}

\hide{We remark that Makanin's algorithm \cite{Makanin:math:1977}
 discards
these cycles in its reduction
 as it only focuses on the satisfiability
of the {\em minimal} solution.
For the general problem,
Jaffa \cite{Jaffar:JACM:1990} extends Manakin's algorithm to derive
a minimal and complete set of all solutions.
As the algorithm in \cite{Jaffar:JACM:1990} also ignores the cycles,
it does not terminate when the set is infinite.
In contrast, {\allSat} takes these cycles into account
for a finite
 representation of a (infinite) set of all solutions.
}

\paragraph{Computing \form{\arith_{xy}} constraint}
Based on the solution trace \form{\sub_{ij}} above, 
 {\deci} 
%
first generates a conjunctive set of constrained Horn clauses to define the
relational assumptions over lengths of \form{x}
and \form{y} in the set of all solutions.
After that it
infers the length constraint as:
 \form{\arith_{xy} {\equiv}\exists i . \sleng{x}{=} 2i{+}1 {\wedge} i{\geq}0 \wedge
\sleng{y}{\geq}0}.
%
Now, the satisfiability of \form{\pure} is
 equi-satisfiable to the following formula:
\form{\pure'{\equiv} (\exists i.~ \sleng{x}{=} 2i{+}1{\wedge} i{\geq}0 \wedge
\sleng{y}{\geq}0) ~{\wedge}~
 (\exists k.~ \sleng{x} {=} 4k{+}3) {\wedge} \sleng{x}{=}2\sleng{y}}.
As \form{\pure'} is unsatisfiable, 
 so is \form{\pure}.


\section{The Representation of All Solutions}\label{sec.all.sol}
In this section, we first present procedure {\allSat}
which constructs a cyclic  reduction tree for
 a conjunction of word equations \form{\ses} (subsection \ref{algo.allsat}).
We presents a fairly complicated
cyclic reduction tree
of \form{\se_c{\equiv}xaby {=} ybax} in subsection \ref{casestudy}.
After that, we describe how to combine the tree
with regular membership predicates \form{\astart} 
(subsection \ref{algo.reg}).
Finally, we discuss the correctness 
 in subsection \ref{algo.correct}. 

\subsection{Constructing Cyclic Reduction Tree}\label{algo.allsat}


 {\allSat} 
 transforms a conjunction
of word equations \form{\ses} into a cyclic reduction tree
\utree{n} which represents
 all its solutions.
This procedure starts with the tree \utree{0}
with only the input \form{\ses} at the root.  
After that, in each iteration it chooses one leaf node
to reduce (using function \code{reduce}) or to make a back-link
(using function  \code{\lb})
until every leaf node is either irreducible
or linked back.
A leaf node is irreducible if it either trivially true
(i.e., \form{w_1{=}w_1{\wedge}...{\wedge}w_i{=}w_i} where
\form{w_1,...,w_i{\in} \alphabet^*})
or trivially false
(i.e., either
  it is of the form: \form{c_1tr_1{=}c_2tr_2{\wedge}\ses} where
\form{c_1}, \form{c_2} are different letters
or
 its over-approximation over the length functions 
 is unsatisfiable).
Function \code{\red} takes a leaf node \form{\ses_i}
as input and produces a set \form{L_i} each element of which is a pair
of a node \form{\ses_{i_j}} and a corresponding substitution \form{\sub_j}
such that \form{\ses_{i_j}{=}\ses_{i}\sub_j}.
For each pair \form{(\ses_{i_j}, {\sub_j}) {\in} L_i},
it adds an new open node \form{\ses_{i_j}} and
 a new edge (\form{\ses_i}, \form{\sub_j}, \form{\ses_{i_j}}).
As a result, \code{\red} extends the current
tree with the new nodes and new edges.
In particular,
function \code{\red} is implemented
as: \form{L_i {=} \bigcup \{\code{\code{matchs}}(\ses_{i_j}) \mid \ses_{i_j} {\in} \comp(\ses_i) \}}
where function \code{matchs} exhaustively matches and subtracts
identical letters and string variables
at the heads of left-hand side (LHS)
and right-hand side (RHS) of each word equation
 using function {\matchse}.
%
In the following, we describe the details of the functions used
by {\allSat}.

\noindent{\bf Matching}
\form{\matchse(\se)} matches two terms
at the heads of LHS and RHS of \form{\se} as follows.
\savespace\[\savespace
\begin{array}{l}
 \form{\matchse(u_1{\cdot}tr_1{=}u_2{\cdot}tr_2)} =
\quad \begin{cases}
   \form{\matchse(tr_1{=}tr_2) } & \quad \text{if } \form{u_1, u_2} \text{ are identical}\\
   u_1{\cdot}tr_1{=}u_2{\cdot}tr_2 & \quad \text{otherwise}
  \end{cases}
\end{array}
\]
where \form{u_1, u_2} are either letters or string variables.




\noindent{\bf Procedure \code{\comp}}
The overall goal of our reduction is to transform
every word equation, say  \form{\se{\equiv}u_1tr_1{=}u_2tr_2}
where \form{\ses_i{=} \se {\wedge} \ses },
into a set of ``smaller" string equation \form{\se_i} such that
if \form{\se} is satisfied,
 \form{\se_i} is a suffix of
\form{\se}.
word equations in a node are reduced in a depth-first manner.
Intuitively, our reduction over the word equation
\form{\se} is based on
the possible arrangements of two
carrier terms, the terms at the 
heads of LHS and RHS of {\se}.
Suppose that \form{\se} is satisfied.
Let \form{l_1}, \form{r_1} are the starting
and ending positions of
 \form{u_1}
in the solution word of \form{\se}.
Similarly, let \form{l_2}, \form{r_2} are the starting
and ending positions of \form{u_1}
in the solution word of \form{\se}.
Obviously, \form{l_1{=}l_2}.
Our reduction, function \form{\comp},
considers all possible arrangements based on
these positions.
For arrangements in one-side (LHS or RHS), it considers the cases:
 \form{l_1{=}r_1}
 (i.e., \form{u_1{=}\semp}), \form{l_1{<}r_1} and
 \form{l_2{=}r_2}
 (i.e., \form{u_2{=}\semp}),
 \form{l_2{<}r_2}.
For arrangements between the two sides, it considers the cases:
\form{r_1{\geq}r_2} and \form{r_2{\geq}r_1}.
In particular,
function \code{\comp}
considers the following two scenarios
of the carrier terms.

\noindent{\bf Case 1:} One term is a letter
and another term is a string variable, e.g. \form{x_1tr_1{=}c_2tr_2}.
\code{\comp} generates the set
\form{L_i} 
as
 \form{L_i {\equiv} \{(\ses_{i_1},{\sub_1}); ({\ses_{i_2}},{\sub_2})\}} where
\begin{itemize}[noitemsep,topsep=0pt]
\item
 1a) \form{\sub_1{=}[\semp/x_1]}
\item
 1b) \form{\sub_2{=}[c_2x_1'/x_1]}, 
 \form{x_1'} is a fresh variable and
 referred as a subterm of \form{x_1}.
\end{itemize}

\noindent{\bf Case 2:} These terms are two different string variables, e.g.
\form{x_1tr_1{=}x_2tr_2}. {\complete} generates the set
\form{L_i} 
as : \form{L_i {\equiv} \{({\ses_{i_1}}, {\sub_1});
  ({\ses_{i_2}}, {\sub_2});  ({\ses_{i_3}}, {\sub_3});
  ({\ses_{i_4}}, {\sub_4})\}} where
 \begin{itemize}[noitemsep,topsep=0pt]
 \item
2a) \form{\sub_1{=}[\semp/x_1]},
\item
 2b) \form{\sub_3{=}[x_2x_1'/x_1]} - \form{x_1'} is a fresh variable and
 referred as a subterm of \form{x_1},
\item 
2c) \form{\sub_2{=}[\semp/x_2]} 
 \item
2d) \form{\sub_4{=}[x_1x_2'/x_2]}, \form{x_2'} is a fresh variable
and referred as a subterm of \form{x_2}.
\end{itemize}

As both Case 2b and Case 2d include the scenario
where \form{x_1{=}x_2}, the reduction tree
generated represents a {\em complete} but {\em not minimal}
set of all solution. 

\noindent{\bf Linking back}
 {\code{\lb}}
links a leaf node \form{\ses_b}
 to an interior node \form{\ses_c} if after some substitution \form{\sub_{cyc}},
 two nodes are identical: \form{\ses_c {\equiv} \ses_b\sub_{cyc}}.
 In addition, for every entry \form{X/X' \in \sub_{cyc}}
where \form{X} and \form{X'} are string variables,
 \form{X'} is a subterm of \form{X}.
 \form{\sub_{cyc}} can be considered
as a permutation function on both \form{\SVar}
and  the
alphabet \form{\alphabet}.
We recap that
we refer to this cycle as a triple \form{\ctree{\ses_c}{\ses_b}{\sub_{cyc}}}
 where \form{\ses_c} is called a companion,
 \form{\ses_b} is called a bud.


\subsection{Cyclic Reduction Tree for \form{\se_c{\equiv}xaby {=} ybax}}\label{casestudy}
\begin{figure}[tb] \savespace \savespace \savespace \savespace
\begin{center}
 \savespace \savespace
\begin{tikzpicture}[node distance=18mm,level 1/.style={sibling distance=32mm},
      level 2/.style={sibling distance=32mm},
                        level distance=22pt, draw]
  \tikzstyle{every state}=[draw,text=black]

\node (A)                    {\textcolor{blue}{$^\bigstar\se_c^\heartsuit$}};
  \node         (B) [below left=46mm and 12mm of A] {\textcolor{orange}{$\se_{1}^\dagger$}};
  \node         (B1) [below right=6mm and 4mm of B] {\underline{$\se_{5}$}};
   \node        (B2) [below left=6mm and 4mm of B] {{$\se_{6}$}};
  \node         (B21) [below right=8mm and 6mm of B2] {{$\se_{13}$}};
   \node         (B22) [below left=8mm and 6mm of B2] {\textcolor{orange}{$\se_{14}^{\dagger}$}};
  \node         (C) [below left=2mm and 28mm of A] {{$\se_{2}$}};
  \node         (F) [below right=46mm and 12mm  of A] {\textcolor{orange}{$\se_{3}^{\ddagger}$}};
  \node         (F1) [below left=6mm and 4mm of F] {\underline{$\se_{9}$}};
   \node         (F2) [below right=6mm and 4mm of F] {{$\se_{10}$}};
  \node         (F21) [below left=8mm and 6mm of F2] {{$\se_{17}$}};
   \node         (F22) [below right=8mm and 6mm of F2] {\textcolor{orange}{$\se_{18}^{\ddagger}$}};
  \node         (G) [below right=2mm and 28mm of A] {{$\se_{4}$}};
  \node         (H1) [below left=4mm and 3mm of G] {\underline{$\se_{11}$}};
   \node         (H2) [below right=4mm and 3mm of G] {{$\se_{12}$}};
   \node         (K1) [below left=6mm and 4mm of H2] {\textcolor{red}{$\se_{19}^\bigtriangleup$}};
   \node         (K2) [below right=6mm and 4mm of H2] {\textcolor{blue}{$\se_{20}^\heartsuit$}};
  \node         (K11) [below left=8mm and 4mm of K1] {{$\se_{23}$}};
   \node         (K12) [below right=8mm and 4mm of K1] {\textcolor{red}{$\se_{24}^\bigtriangleup$}};
  \node         (D) [below right=5mm and 3mm of C] {$\underline{\se_{7}}$};
  \node         (E) [below left=5mm and 3mm of C] {\textcolor{black}{$\se_{8}$}};
   \node         (E1) [below right=7mm and 4mm of E] {\textcolor{red}{$^\bigtriangledown\se_{15}$}};
    \node        (E11) [below right=8mm and 4mm of E1] {{$\se_{21}$}};
   \node         (E12) [below left=8mm and 4mm of E1] {\textcolor{red}{$\se_{22}^\bigtriangledown$}};

   \node         (E2) [below left=7mm and 4mm of E] {\textcolor{blue}{$\se_{16}^\bigstar$}};

  \path (A) edge[->]     node [pos=0.5,right]  {$[\semp/x]$} (B)
            edge[->]     node [pos=0.5,above] {$[yx_1/x]$} (C)
            edge[->]     node [pos=0.5,left] {$[\semp/y]$} (F)
            edge[->]     node [pos=0.5,above] {$[xy_1/y]$} (G)
        (B) edge [->]    node [pos=0.3,right] {$[\semp/y]$} (B1)
            edge [->]    node [pos=0.3,left] {$[ay_3/y]$} (B2)
        (B2) edge [->]    node [pos=0.3,right] {$[\semp/y_3]$} (B21)
            edge [->]    node [pos=0.3,left] {$[by_4/y_3]$} (B22)
        (B22) edge [->,dotted,bend left=85,dotted]  node [pos=0.5,left] {$[y/y_4]$} (B)
        (F) edge [->]    node [pos=0.3,left] {$[\semp/x]$} (F1)
            edge [->]    node [pos=0.3,right] {$[bx_4/x]$} (F2)
        (F2) edge [->]    node [pos=0.3,left] {$[\semp/x_4]$} (F21)
            edge [->]    node [pos=0.3,right] {$[bx_5/x_4]$} (F22)
        (F22) edge [->,dotted,bend right=85,dotted]  node [pos=0.5,right] {$[x/x_5]$} (F)
        (G) edge [->]    node [pos=0.3,left] {$[\semp/y_1]$} (H1)
            edge [->]    node [pos=0.3,right] {$[ay_5/y_1]$} (H2)
        (H2) edge [->]    node [pos=0.3,left] {$[\semp/y_5]$} (K1)
            edge [->]    node [pos=0.05,below] {$[by_6/y_5]$} (K2)
        (K1) edge [->]    node [pos=0.3,left] {$[\semp/x]$} (K11)
            edge [->]    node [pos=0.7,left] {$[ax_6/x]$} (K12)
        (K12) edge [->,dotted,bend right=30,dotted]  node [pos=0.3,right] {$[x/x_6]$} (K1)
        (C) edge [->]    node [pos=0.3,right] {$[\semp/x_1]$} (D)
            edge [->]    node [pos=0.3,left] {$[bx_2/x_1]$} (E)
        (E) edge [->]  node [pos=0.3,right] {$[\semp/x_2]$} (E1)
            edge [->]  node [pos=0.05,below] {$[ax_3/x_2]$} (E2)
        (E1) edge [->]  node [pos=0.3,right] {$[\semp/y]$} (E11)
            edge [->]  node [pos=0.7,right] {$[by_2/y]$} (E12)
        (E12) edge [->,dotted,bend left=30,dotted]  node [pos=0.5,left] {$[y/y_2]$} (E1)
        (E2) edge [->,dotted,bend left=75,dotted]  node [pos=0.5,above] {$[x/x_3]$} (A)
        (K2) edge [->,dotted,bend right=75,dotted]  node [pos=0.5,above] {$[y/y_6]$} (A);
\end{tikzpicture}
 \savespace \saveone 
\caption{Cyclic Reduction Tree $\utree{11}$ for \form{xaby{=}ybax}.}\label{fig.reduction.tree.complex}
\end{center} \savespace 
\end{figure}
 We describe how {\allSat}
 can derive a reduction tree for the word equation:
 \form{\se_c{\equiv}xaby {=} ybax}.
As mentioned before, although
the work presented in
\cite{Plandowski:STOC:2006} 
can derive a graph to finitely represent all solutions
of the word equation \form{\se_c}, the length constraints implied for
variables \form{x} and \form{y} by all solutions of
this  equation can not represented with finitely many
 equations in numeric solvable form.
Our decision procedure
can decide that \form{\pure_c} is satisfiable. Indeed, it
 derives for \form{\se_c}
a reduction tree as presented in Fig. \ref{fig.reduction.tree.complex}
where its nodes are as follows.
\saveone\[\saveone
\begin{array}{llll}
\se_{1}{\equiv}aby{=}yba &\quad \quad
\se_{2}{\equiv}x_1aby {=} bayx_1 &\quad \quad
\se_{3}{\equiv}xab{=}bax &\quad \quad
\se_{4}{\equiv}abxy_1{=}y_1bax \\
\se_{7}{\equiv}aby {=} bay &\quad \quad
\se_{8}{\equiv}x_2aby {=} aybx_2 &\quad \quad
\se_{15}{\equiv}by {=} yb &\quad \quad
\se_{16}{\equiv}x_3aby {=} ybax_3 \\
\se_{21}{\equiv}by {=} yb &\quad \quad
\se_{22}{\equiv}by_2 {=} y_2b &\quad \quad
\se_{5}{\equiv}ab{=}ba &\quad \quad
 \se_{6}{\equiv}bay_3{=}y_3ba  \\
  \se_{13}{\equiv}\semp{=}\semp  &\quad \quad
\se_{14}{\equiv}aby_4{=}y_4ba &\quad \quad
\se_{9}{\equiv}ab{=}ba &\quad \quad
\se_{10}{\equiv}x_4ab{=}abx_4  \\
\se_{17}{\equiv}\semp{=}\semp &\quad \quad
\se_{18}{\equiv}x_5ab{=}bax_5 & \quad \quad
\se_{11}{\equiv}abx{=}bax &\quad \quad
\se_{12}{\equiv}bxay_5{=}y_5bax \\
\se_{19}{\equiv}xa{=}ax &\quad \quad
\se_{20}{\equiv}bxay_6{=}y_6bax  & \quad \quad
\se_{23}{\equiv}\semp{=}\semp &\quad \quad
\se_{24}{\equiv}xa{=}ax 
\end{array}
\]

\subsection{Combining with regular memberships}\label{algo.reg}

%
 We propose to derive
a finite representation of all solutions
 of a conjunction
of word equations and regular expressions.
using procedure \code{widentree}.
Procedure \code{widentree} takes
a pair of a reduction tree \form{\utree{n}} of \form{\ses_0} (generated by {\allSat})
and a conjunction of
 regular expressions \form{\astart} as inputs
and manipulates the reduction tree \form{\utree{n}}  through the following
three steps.
 First, it constructs
a DFA \form{A{=}\dfa{Q}{\alphabet}{\transition}{q_o}{Q_F}} which generates
 the same
 language
with \form{\astart}. Let \form{m} be the number states
in \form{Q} and \form{M{=} m!}. Intuitively,
\form{m{+}1} is the minimal times
of a cycle to obtain the minimal solutions of \form{\ses_0{\wedge}\astart}.
 \form{M} is the periodic of the 
\begin{wrapfigure}{r}{0.57\textwidth} \savespace  \saveone
  \savespace  \savespace \savespace \savespace
\begin{center}
\begin{minipage}{0.57\textwidth}
\scriptsize
\begin{scriptsize}
\begin{tikzpicture}[node distance=18mm,level 1/.style={sibling distance=22mm},
      level 2/.style={sibling distance=25mm},
                        level distance=22pt, draw]
  \tikzstyle{every state}=[draw,text=black]

\node (A)                    {\textcolor{blue}{\small $\se_0{\equiv}abx {=} xba$}};
  \node         (C) [below right=3mm and 0mm of A] {\small{$\se'_{12}{\equiv}bax_1 {=} x_1ba$}};
  \node         (D) [below left=3mm and 0mm of C] {\small ${\se_{21}{\equiv}ba{=}ba}$};
  \node         (E) [below right=3mm and 0mm of C] {\small \textcolor{blue}{$\se_{22}'{\equiv}abx_2 {=} x_2ba$}};
\node (A1)          [below left=2mm and 22mm of E]          {\small \textcolor{blue}{$\se^1_0{\equiv}abx_3 {=} x_3ba$}};
  \node         (C1) [below right=3mm and 0mm of A1] {\small {$\se^{1'}_{12}{\equiv}bax_4 {=} x_4ba$}};
  \node         (D1) [below left=3mm and 0mm of C1] {\small \underline{${\se^1_{21}{\equiv}ba{=}ba}$}};
  \node         (E1) [below right=3mm and 0mm of C1] {\small \textcolor{blue}{$\se^{1'}_{22}{\equiv}abx_5 {=} x_5ba$}};
\node (A2)          [below left=2mm and 22mm of E1]          {\small \textcolor{blue}{$\se^2_0{\equiv}abx_6 {=} x_6ba^\bigstar$}};
  \node         (C2) [below right=3mm and 0mm of A2] {\small {$\se^{2'}_{12}{\equiv}bax_7 {=} x_7ba$}};
  \node         (D2) [below left=3mm and 0mm of C2] {\small \underline{${\se^2_{21}{\equiv}ba{=}ba}$}};
  \node         (E2) [below right=3mm and 0mm of C2] {\small \textcolor{blue}{$\se^{2'}_{22}{\equiv}abx_8 {=} x_8ba^\bigstar$}};

  \path (A) edge[->]     node [pos=0.3,right] {\small \form{[ax_1/x]}} (C)
        (C) edge [->]    node [pos=0.3,left] {\small \form{[\semp/x_1]}} (D)
            edge [->]    node [pos=0.3,right] {\small \form{[bx_2/x_1]}} (E)
        (E) edge [->]  node [pos=0.3,below] {\small $[x_3/x_2]$} (A1)
        (A1) edge[->]     node [pos=0.3,right] {\small  \form{[ax_4/x3]}} (C1)
        (C1) edge [->]    node [pos=0.3,left] {\small \form{[\semp/x_4]}} (D1)
            edge [->]    node [pos=0.3,right] {\small \form{[bx_5/x_4]}} (E1)
        (E1) edge [->]  node [pos=0.3,above] {\small $[x_6/x_5]$} (A2)
        (A2) edge[->]     node [pos=0.3,right] {\small  \form{[ax_7/x_6]}} (C2)
        (C2) edge [->]    node [pos=0.3,left] {\small \form{[\semp/x_7]}} (D2)
             edge [->]    node [pos=0.3,right] {\small \form{[bx_8/x_7]}} (E2)
(E2) edge [->,bend right=45,dotted]  node [pos=0.5,right] {\small $[x_6/x_{8}]$} (A)
;
\end{tikzpicture}
 \savespace 
\end{scriptsize}\savespace \saveone
\end{minipage}
\caption{Extending Tree $\utree{2}$ with \form{x \in a^*}.}\label{fig.unfolding.tree.reg}
 \savespace   \savespace
\end{center}   \savespace 
 \savespace
  \savespace \savespace \saveone
\end{wrapfigure}
sets of all solutions.
 Secondly, it unfolds every cycles \form{\ctree{\ses_c}{\ses_b}{\sub}}
 of \form{\utree{n}}
 \form{m{+}M} times.
It updates \form{\lb} functions by eliminating
   the old back-link
between 
 \form{\ses_b}  and 
 \form{\ses_c}
prior to generating a new back-link between \form{\ses_{b_{m{+}M}}}
and
 \form{\ses_{c_m}}
as well as marking \form{\ses_{b_{m{+}M}}}
as closed. We note
that a solution corresponding to a trace
which visits the companion \form{\ses_{c_m}} \form{l{+}1} times
(i.e., including \form{k} new cycles above)
 has the form:
\form{S \equiv u_1w^{m{+}1{+}lM}u_2}.
 Lastly, it collects label \form{\sub_j} for every
 path \form{(\ses_0, \ses_j)} in the new tree
where 
\form{\ses_0} is the root, \form{\ses_j} is a leaf node
that is neither unsatisfiable nor a bud prior to
evaluating \form{\ses_j}.
From \form{\sub_j}, it generates the following formula: 
\form{\pure_j {\equiv} \bigwedge \{X_i {=} s_i {|} (s_i/X_i) {\in} \sub_j \} {\wedge} \astart}.
\form{\pure_j} is in a {\em straight-line} fragment 
where the satisfiability problem {\strProb} is decidable \cite{Lin:POPL:2016}.

\begin{figure}[tb]
\begin{center}
 \savespace
\begin{tikzpicture}[node distance=18mm,level 1/.style={sibling distance=22mm},
      level 2/.style={sibling distance=25mm},
                        level distance=22pt, draw]
  \tikzstyle{every state}=[draw,text=black]

\node (A)                    {\textcolor{blue}{$\se_0{\equiv}abx {=} xba^\bigstar$}};
  \node         (B) [below left=2mm and 4mm of A] {\underline{$\se_{11}{\equiv}ab{=}ba$}};
  \node         (C) [below right=2mm and 4mm of A] {{$\se_{12}{\equiv}bax_1 {=} x_1ba$}};
  \node         (D) [below left=2mm and 4mm of C] {${\se_{21}{\equiv}\semp{=}\semp}$};
  \node         (E) [below right=2mm and 4mm of C] {\textcolor{blue}{$\se_{22}{\equiv}abx_2 {=} x_2ba^\bigstar$}};

  \path (A) edge[->]     node [pos=0.3,left] {\form{[\semp/x]}} (B)
            edge[->]     node [pos=0.3,right] {\form{[ax_1/x]}} (C)
        (C) edge [->]    node [pos=0.3,left] {\form{[\semp/x_1]}} (D)
            edge [->]    node [pos=0.3,right] {\form{[bx_2/x_1]}} (E)
        (E) edge [->,bend right=30,dotted]  node [pos=0.3,above] {$[x/x_2]$} (A);
\end{tikzpicture}
 \savespace \savespace 
\caption{Reduction Tree $\utree{2}$.}\label{fig.unfolding.tree2}
\end{center} 
  \savespace \saveone
 \savespace  \savespace
\end{figure}
\begin{example}
To illustrate our first decidable fragment, we use the following word equation as a running example:
\form{abx{=}xba} where \form{x} is string variable and \form{a}, \form{b} are letters. This is the first equation in the motivating example (section \ref{sec.mov.example}).
Its reduction tree \form{\utree{2}} is presented in
  Fig. \ref{fig.unfolding.tree2}.
We now illustrate how to use procedure \code{widentree} above to
 extend the tree to represent
all solutions of \form{\pure_1{\equiv}abx{=}xba \wedge x{\in} a^*}.
To do that, \code{widentree} first
derives  for the regular expression \form{x\in a^*}
a DFA
as: \form{A = \dfa{\{q_0\}}{\{a\}}{\{((q_0,a), a)\}}{q_0}{\{q_0\}}},
and then identifies \form{m{=}1} and \form{M{=}m!{=}1}.
Secondly, it clones the cycle of \form{\utree{2}} 
 \form{m+M=1+1=2} more times.
The resulting tree is described in
 Fig. \ref{fig.unfolding.tree.reg}.
Lastly, it discharges the satisfiability of solutions
corresponding to the paths
which start from the root and end at leaf nodes \form{\se_{21}},
\form{\se^1_{21}} or \form{\se^2_{21}}.
 The evaluation is as follows.
\[
\begin{array}{clc}
\text{path} & \text{formula} & \text{outcome}\\
\hline
(\se_0,\se_{21}) & x{=}ax_1 {\wedge} x_1{=}\semp \wedge x{\in}a^* & \sat \\
(\se_0,\se^1_{21}) &
 x=ax_1 {\wedge}  x_1{=}bx_2  {\wedge} x_2{=}x_3  \wedge x_3{=}ax_4
{\wedge} x_4{=}\semp \wedge x{\in}a^*& \unsat \\
(\se_0,\se^2_{21}) &
\begin{array}{l}
 x=ax_1 {\wedge}  x_1{=}bx_2  {\wedge} x_2{=}x_3  \wedge x_3{=}ax_4
{\wedge} x_4{=}bx_5 {\wedge}  \\
 x_5{=}x_6  {\wedge} x_6{=}ax_7
{\wedge} x_7{=}\semp \wedge x{\in}a^*
\end{array}
& \unsat \\
 \end{array}
 \]
\end{example}

 \subsection{Correctness}\label{algo.correct}

In the following, we formalize
the correctness of the proposed procedures
and show the relationship between the derived
reduction tree with \form{\edtl} system \cite{L:Book}.


\begin{proposition}\label{thm.tree}
Suppose that {\allSat} takes a conjunction
 \form{\ses} as input,
and produces a cyclic reduction graph  \form{\utree{n}}
  in a finite time.
  Then, \form{\utree{n}} represents all solutions
of \form{\ses}.
\end{proposition}

\begin{proposition}\label{lemma.edtl.re}
Suppose \form{\astart \equiv X_1 {\in} \regex_1 {\wedge} ...{\wedge}X_n {\in} \regex_n} (\form{X_i {\in} \FV(\ses_0), \forall 1 \leq i \leq n})
 be a conjunction of
 regular memberships and
\form{\utree{n}} be the reduction tree
derived for \form{\ses_0}. Then, \form{\widentree{\utree{n}, \astart}}
produces
a reduction tree representing all solutions of
\form{\ses_0 \wedge \astart}.
\end{proposition}
\repconf{The proof is presented in App. \ref{proof.lemma.edtl.re}.}{}


An {\em interactionless Lindenmayer system} (0L system) \cite{L:Book}
is a parallel rewriting system which was introduced
in 1968 to model the development of multicellular system.
The class of \form{\edtl} languages forms perhaps
the central class in the theory of L systems.
The acronym EDT0L refers to
{\bf E}xtended, {\bf D}eterministic,
 {\bf T}able, {\bf 0} interaction,
and {\bf L}indenmayer. (More discussion on EDT0L language is left in App. \ref{sec.edt0l}.)
In the following, we give a formal definition
of \form{\edtl} system.
\begin{defn}
An \form{\etl} system is a quadruple \form{G{=}\lgrammar{\var}{\alphabet}{\productions}{S}} where
 \form{\var} is a finite nonempty set of nonterminals
 (or variables),
 \form{\alphabet} is a finite set of terminals
and
 disjoint from \form{\var},
 \form{S {\in} \var} is the start variable (or start symbol),
 \form{\productions} is a finite set each element of which
 (called a {\em table}) is a finite binary relation
included in \form{\var \times (\var {\cup} \alphabet)^*}.
It is assumed that \form{\forall \production \in \productions,
\forall x {\in} \var, \exists tr {\in} (\var \cup \alphabet)^*} such that \form{(x,tr) \in \production}.
An \form{\edtl} system is a {\em deterministic} \form{\etl}
system
in which \form{\forall \production {\in} \productions,
\forall x\in \var, \exists! tr {\in} (\var \cup \alphabet)^*} s.t. \form{(x,tr) \in \production}.
\end{defn}
For a production \form{(x{,}tr)} of \form{\production}
in \form{\productions}, we often write:
 \form{x \rightarrow tr}. We also write \form{x \rightarrow_\production tr} for ``\form{x \rightarrow tr} is in \form{\productions}".
Let
\form{G{=}\lgrammar{\var}{\alphabet}{\productions}{S}}
be an \form{\etl} system.
\begin{enumerate}[noitemsep,topsep=0pt]
\item Let \form{x},\form{y\in (\var \cup \alphabet)^*}, and
\form{x} contains \form{k} nonterminals \form{v_1{,}...,v_k} in
 \form{\var}.
 We say that \form{x} directly derives \form{y} (in \form{G}), denoted as \form{x\Rightarrow_G y}, if there is a
 \form{\production \in \productions} such that
 \form{y} is obtained by substituting \form{v_i}
by \form{s_i}, respectively
 for all \form{i \in \{1,...,k\}}, where
 \form{v_1 \rightarrow_\production s_1},
..., \form{v_k \rightarrow_\production s_k}.
In this case, we also write \form{x\Rightarrow_\production y}.
\item Let \form{\Rightarrow^*_G} be the
  reflexive transitive closure 
of the relation \form{\Rightarrow}.
If  \form{x\Rightarrow^*_G y} then we say that \form{x} derives \form{y} (in \form{G}).
\item The language of \form{G}, denoted by \form{\langg{G}}, defined
by \form{\langg{G} = \{w\in \alphabet^* ~|~ S \Rightarrow^*_G w \}}.
\end{enumerate}

A grammar system that is
{\em k-index} is restricted so that, for
every word generated by the grammar, there is some
successful derivation where
at most {\em k} nonterminals
appear in every sentential form
of the derivation \cite{ROZENBERG:1978}.
A system is finite-index if it is \form{k}-index
for some \form{k}.
We use \form{\classlangfin{L}} to denote the class of all \form{L}
languages of finite-index.
\begin{cor}\label{thm.edtf}
A reduction tree
derived by {\allSat} 
  forms a finite-index
\form{\edtl} system.
\end{cor}

\savespace 
\begin{example}\label{edtl.ex}
The tree in the Fig. \ref{fig.unfolding.tree2} above 
forms the following finite-index \form{\edtl}.\\
 \form{G{=}\lgrammar{\{S, x, x_1,x_2\}}{\alphabet}{\{P_1,P_2\}}{S}}
where
\form{\production_1=\{(S, abx), (x,ax_1), (x_1,\semp) \}}
and
\form{\production_2=\{(S, abx), (x,ax_1), (x_1,bx_2), (x_2,x) \}}.
\end{example}


\section{Decision Procedure}\label{sec.deci}

\hide{In this section, 
we present a decision procedure
for the satisfiability problem (subsection \ref{sec.deci.algo}).
Next, we show the correctness of the proposed procedure
(subsection \ref{sec.deci.correct}).
Finally, we identify
a subfragment of \form{\stredtsem} that
syntactically defines a decidable fragment
(subsection \ref{algo.dec.syn}).}


\label{sec.deci.algo}
We present decision procedure \code{\deci}
to handle {\strProb}. 
%
\form{\deci} takes a constraint, 
say \form{\ses{\wedge} \astart{\wedge}\ariths}, as input
and returns {\sat} or {\unsat}. 
It works as follows.
\begin{enumerate}[leftmargin=*] \setlength\itemsep{0em}
\item
 First, it invokes {\allSat} to construct a reduction tree \utree{n}
as a finite representation of all solutions of \form{\ses}.
After that, \utree{n} is post-processed
using procedure {\postpro} as below to explicate all free variables.
This step is critical to the next step. 
\item
 Secondly, 
it uses procedure \code{widentree} to extend \utree{n} 
 with membership predicates \form{\astart}
and obtains
 \utree{n{+}1}.
Note that unsatisfiable nodes in the reduction tree are eliminated.
\item
 Thirdly, it
computes the length constraints
which are precisely implied by
all solutions generated 
through
procedure \extractpres{\utree{n{+}1}}.
These length constrains, say \form{\arith_w}, are computed
as an existentially quantified Presburger formula.
\item
Lastly, \form{\deci}
solves that satisfiability of the conjunction
\form{\arith_w {\wedge}\ariths} which is in the Presburger
arithmetic and decidable \cite{Fischer:SCP:1974}.
\end{enumerate}
%


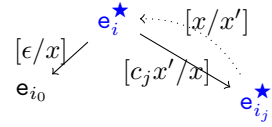
\begin{wrapfigure}{r}{0.33\textwidth} 
\begin{center}
 \savespace \savespace
\begin{tikzpicture}[node distance=18mm,level 1/.style={sibling distance=22mm},
      level 2/.style={sibling distance=25mm},
                        level distance=22pt, draw]
  \tikzstyle{every state}=[draw,text=black]

\node (A)                    {\textcolor{blue}{$\se_i^\bigstar$}};
  \node         (B) [below left=4mm and 4mm of A] {{$\se_{i_0}$}};
  \node         (C2) [below right= 4mm and 12mm of A] {\textcolor{blue}{$\se_{i_j}^\bigstar$}};

  \path (A) edge[->]     node [pos=0.3,left] {\form{[\semp/x]}} (B)
            edge[->]     node [pos=0.3,below] {\form{[c_jx'/x]}} (C2)
        (C2) edge [->,bend right=30,dotted]  node [pos=0.3,above] {$[x/x']$} (A)
;
\end{tikzpicture}
 \savespace \savespace 
\caption{Free Variable $x$.}\label{fig.free.tree}
\end{center}  \saveone  \savespace \savespace \savespace \savespace
\end{wrapfigure}
\paragraph{Post-Processing}
Given a path from the root \form{\se_0}
 to a satisfiable leaf node \form{\se_i},
a variable \form{x}
 appearing in this path
is called {\em free} if it has not been reduced yet.
This means \form{x} can be assigned any value
in \form{\alphabet^*} in a solution.
Procedure {\postpro} aims to
replace a free variable by a sub-tree
which represents
 for arbitrary values
in \form{\alphabet^*}.
The sub-tree is presented in Fig. \ref{fig.free.tree}.
This tree has a {\em base} leaf node
  (with substitution \form{[\semp/x]})
and \form{k} cycles (\form{k} is the size of the
alphabet \form{\alphabet}) one of which represents for
a letter \form{c_i \in \alphabet}.
If a satisfiable leaf node has more than one free variable,
each variable is replaced by such sub-tree
and these sub-trees are connected together at base nodes.





\paragraph{Correctness}
\label{sec.deci.correct}
The correctness of 
step 1 and step 2 have been shown
in the previous section.
%
%
%
%
Thus, the remaining tasks to show {\deci} is a decision procedure in a fragment
 are the termination of {\allSat}
as well as 
the decidability of
\extractpres{\utree{n{+}1}}.

\section{\form{\stredt} Decidable Fragment}\label{sec.quad}

Computing length constraint in this fragment is based on
   Parikh's Theorem \cite{Parikh:JACM:1966}, one of the most
  celebrated theorem in automata theory.
The Parikh image  (a.k.a. letter-counts) of a word over a given
alphabet counts the number of occurrences of each symbol
in the word without regard to their order.
The Parikh image of a language is
the set of  Parikh images of the words in the language.
A language is Parikh-definable if its
Parikh image precisely coincides with semilinear sets
which, in turn, can be
 computed as a 
 Presburger formula.
In particular,
Parikh's Theorem \cite{Parikh:JACM:1966}
 states that
 context-free languages (and regular languages, of course)
are Parikh-definable.
In fact, given a context-free grammar, we can
compute its Parikh image in polynomial time \cite{Verma2005,Esparza:IPL:2011}. 
Moreover, the authors in \cite{ROZENBERG:1978} show that
 finite-index EDT0L languages \cite{L:Book}
are also Parikh-definable.
In our work, we use \form{\Par{L}} to denote
the Parikh images computed for the language \form{L}.


Given a constraint, say \form{ 
 \ses{\wedge} \astart{\wedge}\pure}, is said to be in the fragment if
the following two conditions hold.
First, {\allSat} terminates on \form{\ses}.
Secondly, \form{\pure \equiv \arith_1{\wedge}..{\wedge}\arith_n}
where \form{\FV(\arith_i)} contains at most one string length \form{\forall i\in\{1...n\}}.
%
By the first condition, {\deci} can derive for
\form{\ses} a finite-index \form{\edtl}
system (Corollary \ref{thm.edtf}).
Moreover, finite-index \form{\edtl}  can be translated into a Parikh-equivalent
DFA (by Parikh's Theorem \cite{Parikh:JACM:1966,ROZENBERG:1978}).
This means length of each string variable in the set of all solutions
can be computed as a DFA.
By the second condition, each constraint \form{\arith_1}
is based on the length of one string variable. Hence, this constraint
can be translated into another DFA.
As regular languages are closed under intersection.
Therefore, the satisfiability of \form{\pure} is decidable.

{\deci} uses
\extractpres{\utree{n{+}1}} to compute the length constraints
represented for all solutions of \form{\ses{\wedge} \astart}
as follows. 
Firstly, it transforms \form{\utree{n{+}1}} 
into a
finite-index \form{\edtl}
system. 
Secondly, it transforms the \form{\edtl} grammar into
a Parikh-equivalent CFG \form{G} (see \cite{ROZENBERG:1978}).
Lastly, it computes the length constraints \form{\arith_w}
for every string variables
as: \form{\arith_w{\equiv} \bigwedge \{\Par{\langg{G_x}} \mid x \in \FV(\ses{\wedge} \astart) \} }.




\subsection{Parikh Image of {\cfg}}
In order to infer the Parikh image for a given {\cfg},
 we first transform the CFG into a
Parikh equivalent communication-free Petri net
 and then compute the Parikh image of the communication-free Petri net
\cite{Verma2005}. 
The correctness was
 presented in \cite{Esparza:FCT:1995,Seidl:ICALP:2004,Verma2005}.
Procedure \form{Par} takes a {\cfg}
\form{G{=}\cfgrammar{\var}{\alphabet}{\production}{s_0}}
as input and produces a Presburger formula
to represents the Parihk image of
all words derived from the start symbol \form{s_0}.
In particular, it first transforms the {\cfg}
into a communication-free Petri net
and then generates a Presburger formula \form{\arith_G} for
this net.

 A net \form{N} is a quadruple
\form{N{=}\netgrammar{S}{T}{W}{s_0}} where
\form{S} is a set of places, \form{T} is a
set of transitions,
\form{W} is a weight function:
 \form{(S \times T) \cup (T \times S) \rightarrow \mathbb{N}},
and \form{s_0} is the start place in the net.
If \form{W(x, y) {>} 0}, there is an edge from
\form{x} to \form{y} of weight
\form{W(x, y)}.
A net is communication-free if for each
transition \form{t} there is at most one place
\form{s} with \form{W(s, t) > 0}
and furthermore \form{W(s, t) = 1}.
A marking \form{M}, a function \form{S \rightarrow \mathbb{N}},
 associates a number of tokens
with each place.
A communication-free Petri net is a pair \form{(N, M)}
where \form{N} is a communication-free net
and \form{M} is a marking.

The {\cfg} \form{G} 
is transformed into a  communication-free Petri net
\form{(N_G, M_G)} as: \form{N_G{=}\netgrammar{\var {\cup} \alphabet}{\production}{W}{s_0}}.
If \form{A {\rightarrow} s} is a production \form{p \in \production}
then \form{W(A, p) {=} 1}
and \form{W(B, p)} is the number occurrences of \form{B}
in \form{s}, for each \form{B {\in} \var {\cup} \alphabet}.
Finally, \form{M_G(s_0) {=} 1} and \form{M_G(X) {=} 0}
for all other \form{X {\in} \var {\cup} \alphabet}
and \form{X {\not=} s_0}.
Let \form{x_c} be a new integer variable for each letter \form{c {\in} \alphabet},
\form{y_p} be a new integer variable for each rule \form{p {\in} \production},
and \form{z_s} be a new integer variable for each symbol
 \form{s {\in} \var \cup \alphabet}.
We assume that we have \form{m}
variables \form{y_{p_1},..,y_{p_m}} and 
\form{n} variables \form{z_{s_1},..,z_{s_n}}.
We note that \form{x_c} is used to count the number occurrences
of the letter \form{c{\in}\alphabet} in a word derived by the grammar \form{G}.
The output \form{\arith_G} is generated through the following two steps.
Firstly, the procedure generates
a quantifier-free Presburger formula \form{\arith_{count}}
which constrains the occurrences of letters in words derived
by  the grammar \form{G}. In particular, \form{\arith_{count}}
 is a conjunction of
the four following kinds of subformulas.
\begin{itemize}[noitemsep,topsep=0pt]
\item \form{x_c {\geq} 0} for all \form{c {\in} \alphabet}.
\item For each \form{X {\in} \var},
let \form{p_1}, ..., \form{p_k} be all productions
which \form{X} is on the left-hand side.
And we 
recap \form{W(X,p)} denotes the number occurrences of
\form{X} on the right-hand side of the production rule
\form{p}.
Then, \form{\arith_{count}} contains the following conjunct:
\savespace\[\savespace
\form{M_G(X) + \Sigma_{p \in \production} W(X,p)y_{p} - \Sigma_{i{=}1}^ky_{p_i} = 0}
\]
\item For each \form{c \in \alphabet},
  \form{\arith_{count}} contains the following conjuncts:
\savespace\[\savespace
\form{x_c = \Sigma_{p \in \production} W(c,p)y_{p} \wedge (x_c = 0 \vee z_c>0)}
\]
\item For each \form{s \in \var \cup \alphabet}, let
\form{p_1},...,\form{p_l} be the productions
where \form{s} is on the right-hand side
and \form{X_1},...\form{X_l} are their corresponding left-hand
sides.
Then,  \form{\arith_{count}} contains the following conjunct:
\form{(z_s{=}0
~ {\vee}~ \bigvee_{i{=}1}^l(z_s=z_{X_i} {+} 1 {\wedge} y_{p_i}{>}0 {\wedge} z_{X_i}{>}0)}.
If one of the \form{X_i} is the start symbol \form{s_0},
the corresponding disjunct is replaced by \form{z_s{=}1 {\wedge} y_{p_i}{>}0}.
\end{itemize}

Secondly, \form{\arith_G} is generated as:
\form{\arith_G\equiv 
\exists {y_{p_1},..,y_{p_m}} ,z_{s_1},..,z_{s_n}.
\sleng{s_0}{=}\Sigma_{c \in \alphabet}x_c \wedge \arith_{count}}.

The correctness of \form{Par} immediately follows the following theorem.
\begin{theorem}[\cite{Seidl:ICALP:2004,Verma2005}]
Given a {\cfg} \form{G}, one can
compute an existential Presburger formula \form{\arith {\equiv} \exists {y_{p_1},..,y_{p_m}} ,z_{s_1},..,z_{s_n}. \arith_{count}}
for the Parikh image of \form{\classlang{G}} in linear time.
\end{theorem}

\begin{example}
For the \form{\edtl} in Ex. \ref{edtl.ex}, we generate the following
 Parikh-equivalent CFG \form{G_1} \form{\cfgrammar{\var_1}{\alphabet}{\production_1}{S_1}} where
 the start symbol \form{S_1} is fresh,
\form{\var_1{=}\{S_1,x,x_1,x_2,x_3\}} and
\form{P_1 {\equiv} \{(S_1, abx), (x,  ax_1),(x_1,  bx_2),
(x_2,  x), (x, x_3), (x_3, ax_1), (x_1, \semp) \}}.

Next, we show how to compute \form{Par(\classlang{G_{1_x}})},
Parikh image of {\cfg} \form{G_{1_x}}. 
Let \form{x_a} and
\form{x_b} be integer variables which count the occurrences
of letters \form{a} and \form{b}, resp., of every word.
Let \form{y_1}, \form{y_2},..., \form{y_7} be integer variables
representing for the each production in \form{\production_1}
following the 
left-right order.
And let \form{z_a}, \form{z_b}, \form{z_{S_1}}, \form{z_x},
\form{z_{x_1}}, \form{z_{x_2}} and \form{z_{x_3}} be integer
variables which reflect the distance of the corresponding
symbols to the start symbol \form{x} in a spanning tree
on the subgraph of the transformed net induced by those \form{p}
 with \form{y_p {>} 0}.
The first kind of conjuncts in \form{\arith_{count}} is:
\form{x_a{\geq}0 {\wedge} x_b{\geq}0}.
The second 
is:
\saveone\[\saveone
\begin{array}{l | l}
\begin{array}{cl}
\text{Variable} &\quad conjunct \\
x & 1 + (y_4 + y_1) - (y_2 + y_5) = 0\\
S_1 & 0 + 0 - y_1 = 0\\
x_1 & 0 + (y_2 + y_6) - (y_3 + y_7) = 0 \\
\end{array}
 &
\begin{array}{cl}
\text{Variable} &\quad conjunct \\
x_2 & 0 + y_3 - y_4 = 0\\
x_3 & 0 + y_5 - y_6 = 0\\
\end{array}
\end{array}
\]
The third kind of conjuncts in \form{\arith_{count}}
corresponding to letter \form{a} and \form{b} is:
\form{x_a {=} y_1 {+} y_2 {+} y_6 \wedge (x_a{=}0 {\vee} z_a{>}0)} and
 \form{x_b {=} y_1 {+} y_3 {\wedge} (x_b{=}0 {\vee} z_b{>}0)}, respectively.
The fourth 
is as follows.
\savespace\[\savespace
\begin{array}{ll}
x \quad& z_x=0 \vee (z_x=z_{x_2} + 1 \wedge y_4>0 \wedge z_{x_2} > 0)
\vee (z_x=z_{S_1} + 1 \wedge y_1>0 \wedge z_{S_1} > 0)\\
S_1 & z_{S_1} = 0\\
x_1 & z_{x_1} > 0 \vee (z_{x_1} = 1 \wedge y_2 > 0)
 \vee (z_{x_1}=z_{x_3} + 1 \wedge y_6>0 \wedge z_{x_3} > 0) \\
x_2 &  z_{x_2} > 0 \vee (z_{x_2}=z_{x_1} + 1 \wedge y_3>0 \wedge z_{x_1} > 0)\\
x_3 &  z_{x_3} > 0 \vee (z_{x_3} = 1 \wedge y_5 > 0)\\
a & z_{a} {>} 0 \vee (z_{a}{=}z_{S_1} {+} 1 \wedge y_1{>}0 \wedge z_{S_1}{>} 0)
\vee  (z_{a} {=} 1 {\wedge} y_2 {>} 0) \vee
(z_{a}{=}z_{x_3} {+} 1 \wedge y_6{>}0 \wedge z_{a} {>} 0)\\
b & z_{b} > 0 \vee (z_{b}=z_{S} + 1 \wedge y_1>0 \wedge z_{S_1} > 0)
\vee (z_{a}=z_{x_1} + 1 \wedge y_3>0 \wedge z_{a} > 0)\\
\end{array}
\]
Then, the length constraint of \form{x} is inferred as:
\savespace\[\savespace
\begin{array}{lcl}
\form{\arith_{{G_1}_x}} & {\equiv} & \exists y_1,..,y_7, z_a, z_b,z_x,z_{S_1},
z_{x_1}, z_{x_2}, z_{x_3}. \sleng{x}{=} x_a {+} x_b {\wedge} \arith_{count} \\
 & {\equiv} & \exists y_1,..,y_7, z_a, z_b,z_x,z_{S_1},
z_{x_1}, z_{x_2}, z_{x_3} . \sleng{x}{=} 2y_3{+}1 {\wedge}  x_a{=}y_3{+}1 {\wedge}
x_b{=}y_3 {\wedge}
 \arith_{count}\\
\end{array}
\]
\end{example}


 \subsection{\form{\stredt}:
  A Syntactic Decidable Fragment} \label{algo.dec.syn}

\begin{defn}[\form{\stredt} Formulas]
 \form{
   \ses{\wedge} \astart{\wedge}\arith_1{\wedge}..{\wedge}\arith_n}
is called in 
fragment \form{\stredt}
if \form{\ses} is a quadratic system
and \form{\FV(\arith_i)} contains at most one string length \form{\forall i\in\{1...n\}}.
\end{defn}


The decidability relies on the termination
of  {\allSat} over quadratic systems.

\begin{proposition}\label{lemma.allsat.quad}
{\allSat} runs in factorial time in the worst case for
quadratic systems.
\end{proposition}
\repconf{The proof is presented in App. \ref{proof.lemma.allsat.quad}.}{}


Let {\strProb}[\form{\stredt}]
be the satisfiability problem in this fragment.
 The following theorem
 immediately follows from Proposition \ref{lemma.allsat.quad},
Corollary \ref{thm.edtf}, Parikh image of finite-index \form{\edtl} systems
\cite{ROZENBERG:1978}. 
\begin{theorem}\label{thm.edtl.syn}
{\strProb}[\form{\stredt}] is decidable.
\end{theorem}

\section{\form{\strflat} Decidable Fragment}\label{sec.flat}
We 
first describe \form{\strflatsem} fragment through
 a semantic restriction
and then show the computation of the length constraints. 
After that, 
 we syntactically define \form{\strcf}. 

\savespace
\begin{defn}
The normalized formula
 \form{\ses{\wedge} \astart{\wedge}\ariths}
is called in the \form{\strcfsem} 
fragment
if {\allSat} takes \form{\ses} as input,
and produces a tree  \form{\utree{n}}
  in a finite time. Furthermore,
for every cycle \form{\ctree{\ses_c}{\ses_b}{\sub_{cyc}}}
 of \form{\utree{n}}, 
  every label along the path (\form{\ses_c}, \form{\ses_b})
 is of the form:
\form{[cY/X]}
where \form{X}, \form{Y} 
are string
variables
and 
 \form{c} is
a letter.
\savespace\end{defn}
This restriction implies that \form{\utree{n}} does not contain
any nested cycles. 
We refer such \form{\utree{n}} as a {\em flat(able)} tree.
It further implies that
\form{\sub_{cyc}} is of the form
  \form{\sub_{cyc}\equiv [X_1/X'_{1}, ...,X_k/X'_k]}
and \form{X'_j} is a (direct or indirect) subterm of 
\form{X_j} for all \form{j \in \{1...k\}}.
We refer the variables \form{X_j} for all \form{j \in \{1...k\}} as  extensible variables
and such cycle as \form{\vtree{\ses_c}{\ses_b}{\sub_{cyc}}{[X_1,...,X_k]}}.



\paragraph{Procedure \form{extract\_pres}}\label{sec.flat.length}
From
 a reduction tree,
we propose 
 to extract a system of inductive predicates which precisely
capture
 the length constraints of string variables.

We assume that the system  \form{\PName}
includes \form{n} {\em unknown} (a.k.a. uninterpreted) predicates
and \form{\PName} is defined by a set of constrained Horn clauses.
We notice that, as shown in
 Fig. \ref{prm.spec.fig},
inductive predicates are restricted within arithmetic domain.
Every clause is of the form:
 \form{\phi_{i_j} \imply \seppredF{\code{P_i}}{\setvars{v}_i}}
where \form{\seppredF{\code{P_i}}{\setvars{v}_i}}
is the head and \form{\phi_{i_j}} is the body.
A clause without head is called a query.
A formula without any inductive predicate
is referred as a {\em base} formula
and denoted as \form{\base}.
We now introduce \form{\Gamma} to denote an interpretation over
unknown predicates such that for every \form{\code{P_i} \in \PName},
 \form{\Gamma(\code{P_i}(\setvars{v}_i)) \equiv \base_i}.
We use \form{\phi(\Gamma)} to denote a formula obtained by replacing
all unknown predicates in \form{\phi} with their definitions in \form{\Gamma}.
We say a clause \form{\phi_b \imply \phi_h} satisfies
 if there exists \form{\Gamma} and  for all stacks \form{\sstack {\in} \Stack},
 we have \form{\sstack \models \phi_b(\Gamma)}
implies  \form{\sstack \models \phi_h(\Gamma)}.
A conjunctive set of Horn clauses (CHC for short), denoted by \form{\Horn},
is satisfied if every constraints in \form{\Horn} is satisfied
under the same interpretation of unknown predicates.

We maintain a one to one function
that maps every string variable \form{x {\in} \SVar}
to its respective length variable \form{n_x {\in} \IVar}.
We further distinguish \form{\SVar} into two disjoint sets:
\form{\SGVar} a set of
 global variables and \form{\SLVar} a set of local
(existential)  variables.
While \form{\SGVar} includes those variables from the root
of a reduction tree, \form{\SLVar} includes those fresh variables
generated by {\allSat}.
 Given
a tree \form{\utree{n{+}1}} \form{(V, E, \backfun)}
(where \form{\ses_0 {\in} V} be the root of the tree)
deduced from an input \form{\ses_0 {\wedge}\astart},
we generate a system of inductive predicates
and CHC \form{\Horn} as follows.
\begin{enumerate}
\item For every node \form{\ses_i{\in}V} s.t.
\form{\setvars{v}_i{=}\FV(\ses_i){\neq}\emptyset},
we generate an inductive predicate \form{\seppredF{\code{P_i}}{\setvars{v}_i}}.
\item For every edge \form{(\ses_i, \sub, \ses_j) {\in} E}, 
\form{\setvars{v}_i{=}\FV(\ses_i){\neq}\emptyset},
\form{\setvars{v}_j{=}\FV(\ses_j)},
\form{\setvars{w}_j{=} \setvars{v}_j {\cap} \SLVar}, we generate
the clause:
 \form{\exists \setvars{w}_j.~ \code{gen}(\sub) \wedge
\seppredF{\code{P_j}}{\setvars{v}_j} \imply \seppredF{\code{P_i}}{\setvars{v}_i}
  }
 where \code{gen}(\sub) is defined
as:
\saveone\[\saveone
\begin{array}{l}
 \form{\code{gen}(\sub)} ==
\quad \begin{cases}
   n_x{=}0 & \quad \text{if } \sub{\equiv} [\semp/x]\\
   n_x{=}n_y{+}1 & \quad \text{if } \sub{\equiv}[cy/x] \\
   n_x{=}n_y{+}n_z & \quad \text{if } \sub{\equiv}[yz/x] \\
  \end{cases}
\end{array}
\]
\item For every cycle \form{\ctree{\ses_c}{\ses_b}{\sub_{cyc}} {\in} \backfun}, we generate the following clause:
\savespace\[\savespace
\form{\bigwedge \{{v}_{b_i} {=} {v}_{c_i} \mid [{v}_{c_i}/{v}_{b_i}] \in \sub_{cyc}\}
~ {\wedge}~\seppredF{\code{P_c}}{\setvars{v}_c} \imply \seppredF{\code{P_b}}{\setvars{v}_b} }
\]
\end{enumerate}
The length constraint of all solutions
of \form{\ses_0 {\wedge}\astart} is
captured by the query:
 \form{\seppredF{\code{P_0}}{\FV(\ses_0)}}.


In the following, we show that if \form{\utree{n}} is a flat tree,
the satisfiability of
the generated CHC is decidable.
This decidability relies on the decidability of inductive
predicates in DPI fragment which is presented in \cite{Makoto:APLAS:2016}. In particular,
a system of inductive predicates is in DPI fragment
if every predicate \code{P} is defined as follows.
Either it is constrained by one base clause as:
 \form{\base \imply \seppredF{\code{P}}{\setvars{v}}}
or it is defined by two clauses as:
\saveone\[\saveone
\begin{array}{cc}
\form{\base_1{\wedge}..{\wedge}\base_m \imply \seppredF{\code{P}}{\setvars{v}}} & \qquad
\exists \setvars{w}. \bigwedge \{ \setvars{v}_i \underline{+} \setvars{t}_i {=}k \}
{\wedge} \seppredF{\code{P}}{\setvars{t}}
\imply \seppredF{\code{P}}{\setvars{v}}
\end{array}
\]
where \form{\FV(\base_j) \in \setvars{v}}
 (for all \form{i \in 1..m}) and
has at most one variable;
\form{\setvars{t} \subseteq \setvars{v} \cup \setvars{w}},
 \form{\setvars{v}_i} is the variable at $i^{th}$ position
of the sequence \form{\setvars{v}},
and \form{k \in \mathbb{Z}}.

To solve the generated clauses \form{\Horn}, we infer definitions
for the unknown predicates in a bottom-up manner.
Under assumption that \form{\utree{n}} does not contain any
mutual cycles, all mutual recursions can be eliminated
and predicates are in the DPI fragment.
\begin{proposition}\label{prop.flat.length}
The length
 constraint implied by
a flat tree 
is 
Presburger-definable.
\end{proposition}

\begin{example}[Motivating Example Revisited]
We generate the following CHC for the tree $\utree{3}$
in Fig. \ref{fig.unfolding.tree}.
\saveone\[\saveone
\begin{array}{ll}
\begin{array}{lll}
\exists n_{x_1}. n_{x}{=}n_{x_1}{+}1 {\wedge}
 \seppredF{\code{P_{12}}}{n_{x_1}{,} n_y} &\imply& \seppredF{\code{P_{0}}}{n_{x}{,} n_y} \\
n_{x_1}{=}0 {\wedge}
 \seppredF{\code{P_{21}}}{ n_y} &\imply& \seppredF{\code{P_{12}}}{n_{x_1}{,} n_y} \\
\exists n_{x_2}. n_{x_1}{=} n_{x_2}{+}1 {\wedge}
 \seppredF{\code{P_{22}}}{n_{x_2}, n_y} &\imply& \seppredF{\code{P_{12}}}{n_{x_1}{,} n_y} \\
 n_{x_2}{=} n_{x} {\wedge}
 \seppredF{\code{P_{0}}}{n_{x}, n_y} &\imply& \seppredF{\code{P_{22}}}{n_{x_2}{,} n_y} \\
\end{array} &\quad
\begin{array}{l}
\begin{array}{lll}
n_{y}{=}0  &\imply& \seppredF{\code{P_{21}}}{n_y} \\
\exists n_{y_1}. n_{y}{=}n_{y_1}{+}1 {\wedge} \seppredF{\code{P_{32}}}{n_{y_1}} &\imply& \seppredF{\code{P_{21}}}{n_y} \\
n_{y_1}{=}n_{y} {\wedge} \seppredF{\code{P_{21}}}{n_{y}}  &\imply& \seppredF{\code{P_{32}}}{n_{y_1}} \\
\end{array}\\
\seppredF{\code{P_{0}}}{n_{x}{,} n_y} {\wedge}
(\exists k . n_{x} {=} 4k{+}3) {\wedge}
  n_{x} {=} 2 n_{y}
\end{array}
\end{array}
\]
After eliminating the mutual recursion,
predicate \code{P_{21}} is in the DPI fragment
and generated a definitions as:
 \form{\seppredF{\code{P_{21}}}{n_{y}} \equiv n_{y}{\geq}0}.
 Similarly,
after substituting the definition of \code{P_{21}} into
the remaining clauses and eliminating the mutual recursion,
predicate \code{P_{0}} is in the DPI fragment
and generated a definitions as:
 \form{\seppredF{\code{P_{0}}}{n_x{,}n_{y}} \equiv \exists i. n_x{=}2i{+}1 {\wedge} n_{y}{\geq}0}.
\end{example}


\paragraph{\form{\strflat} Decidable Fragment}\label{sec.flat.syn}

A quadratic word equation is called {\em regular}
if it is either acyclic or of the form \form{Xw_1=w_2X}
where \form{X} is a string variable and
 \form{w_1, w_2 \in \alphabet^*}.
A quadratic word equation
 is called \form{n} {\em phased-regular}
if it is of the form: \form{s_1{\cdot}...{\cdot}s_n=t_1{\cdot}...{\cdot}t_n} where \form{s_i{=}t_i} is a regular equation for
all \form{i \in \{1...n\}}. 

\begin{defn}[\form{\strflat} Formulas]
 \form{\pure{\equiv} \ses{\wedge} \astart{\wedge}\ariths}
is called in the \form{\strflat} 
fragment
if either \form{\ses} is both quadratic
and phased-regular or \form{\ses} is in SL fragment.
\end{defn}


\begin{proposition}
\label{lemma.cfg.phased.twisted.prop}
 {\allSat} constructs a flat tree for
a \form{\strflat} constraint in linear time.
\end{proposition}

Let {\strProb}[\form{\strflat}]
be the satisfiability problem in this fragment.
\begin{theorem}
{\strProb}[\form{\strflat}] is decidable.
\end{theorem}

\section{Implementation and Evaluation} \label{sec.impl}
We have implemented a prototype for {\deci}, using OCaml,
to handle
the satisfiability problem in theory of word equations
and  length constraints over the Presburger arithmetic.
It takes a formula in SMT-LIB format version as input
and produces {\sat} or {\unsat} as output.
For the problem beyond the decidable fragments,
{\allSat} may not terminate
and {\deci} may return
{\unknown}.
Our SMT-LIB parser is based on the open source \cite{fontend}.
We made use of
 Z3 \cite{TACAS08:Moura} as a back-end SMT solver for the linear arithmetic.


\paragraph{Evaluation}
As noted in \cite{Ganesh:HVF:2012,Tao:POPL:2018}, all
constraints in the standard Kaluza benchmarks \cite{Saxena:SP:2010}
with 50,000+ test cases generated by symbolic execution on
JavaScript applications 
satisfy the straight-line conditions.
Therefore, it could not be used to evaluate
our  proposal that focuses on cyclic constraints.
We have generated and experimented {\deci} over
 a new set of 600 hand-drafted benchmarks each of which is a 
{\em phased-regular} constraint in
the proposed decidable fragment \form{\strflat}.
The set of benchmarks includes 
298 satisfiable queries  and 302 unsatisfiable queries. 
Each benchmark has from one to three phases.
Each phase is in the form of either \form{xaby{=}ybax} or \form{xab{=}bax}
in the case of satisfiable constraints (files quad{-}*odd\_number{-}*)
and \form{xaay{=}ybax} or \form{xaa{=}bax} in the case of
 unsatisfiable constraints (files quad{-}*even\_number{-}*)
where \form{x}, \form{y} are string variables and \form{a}, \form{b}
are letters.
We have also compared {\deci}
against existing state-of-the-art string solvers:
Z3-str2 \cite{Zheng:FSE:2013,Zheng:CAV:2015}, Z3str3 \cite{Berzish:FMCAD:2017}, CVC4 \cite{Liang:FMSD:2016},
 S3P \cite{Trinh:CAV:2016}, 
 Norn \cite{Abdulla:CAV:2014,Abdulla:CAV:2015}
and Trau \cite{Parosh:PLDI:2017}.
All experiments were performed on an Intel Core i7 3.6Gh with 12GB RAM.

\begin{table}[t]
\caption{Experimental Results}
\label{tbl:expr}\savespace
\begin{center}
\begin{tabular}[t]{|c | c |c|c|c|c|c|c|c|} 
\hline
 & {\#}$\surd$\sat & {\#}$\surd$\unsat &
 {\#}\ding{55}\sat & {\#}\ding{55}\unsat &
 {\#}\unknown & {\#}timeout & $ERR$ & Time \\
\hline
\rowcolor{Gray}Trau \cite{TRAU:Jan2018} & 0 & 302 & 236 & 0 & 0 & 0 & 62 &  37s\\
 S3P \cite{S3P:Jan2018}& 55 & 110 & 1 
 & 0 & 100 & 253 & 81 & 801m55s\\
 \rowcolor{Gray} CVC4 \cite{CVC4:Jan2018}& 120 & 143 & 0
 & 69  & 0 & 268  & 0  & 795m49s \\
Norn \cite{Norn:Jan2018} & 67 & 98 & 0
 & 3 
 &  432 & 0  &  0 & 336m20s \\
\rowcolor{Gray}Z3str3 \cite{Z3str3:Jan2018} & 69 & 102
 & 0 
  & 0
 & 292 & 24 & 113  & 77m4s\\
 Z3str2 \cite{Zheng:CAV:2015} & 136 & 66
 & 0 
  & 0
 & 380 & 18 & 0  & 54m35s\\
\rowcolor{Gray}{\deci} & 298 & 302 & 0 & 0 & 0 & 0 & 0 & 18m58s\\
\hline
\end{tabular}
\end{center}\savespace \savespace \savespace \savespace
\end{table}

The experiments are shown in Table \ref{tbl:expr}.
The first column shows the solvers.
The column {\#$\surd$\sat} (resp., {\#$\surd$\unsat})
indicates the number of benchmarks for which
the solvers decided {\sat} (resp., {\unsat}) correctly.
The column {\#\ding{55}\sat} (resp., {\#\ding{55}\unsat})
indicates the number of benchmarks for which
the solvers decided {\unsat} on satisfiable queries (resp., {\sat} on unsatisfiable queries).
The column {\#\unknown} indicates  the number of benchmarks for which
the solvers returned unknown,
{\em timeout} 
for which
the solvers were unable to decide within 
180 seconds,
{\em ERR} for 
internal errors.
The column {\em Time} gives CPU running time 
({\em m} for minutes
and {\em s} for seconds) taken by the solvers.

The experimental results show that
among the existing techniques that deal with cyclic scenarios,
the 
 method
presented by Z3-str2 performed the most effectively and efficiently.
It could detect the overlapping variables in 380 problems (63.3\%) 
without any wrong outcomes in a short running time.
Moreover, it could decide 202 problems (33.7\%) correctly.
CVC4 produced very high number of correct outcome (43.8\% - 263/600).
However, it returned both false positives and false negatives.
Finally, non-progressing detection method in S3P worked not very well.
It detected non-progressing reasoning in only 98 problems (16.3\%) but produced
false negatives and high number of timeouts and internal errors (crashes).
Surprisingly, Norn performed really well. It could detect the highest number
of the cyclic reasoning (432 problems - 72\%).
Trau eventually returned either crashes or {\unsat} for all benchmarks.
%
The results also show that
 {\deci} was both effective
and efficient on
these benchmarks. It decided correctly all queries within a short
running time. These results are encouraging us to
extend the proposed cyclic proof system to support
inductive reasoning over
other string operations
(like \code{replaceAll}).

To highlight our contribution, we revisit the problem
\form{\se_{c}{\equiv}~ xaay {=} ybax} (highlighted in Sect. \ref{sec.intro})
which is contained in
file \code{quad{-}004{-}2-unsat} of the benchmarks.
{\deci} generates a cyclic proof for  \form{\se_{c}}
with the base case \form{\se_{c}^1 {\vee} \se_{c}^2} where
\form{\se_{c}^1{\equiv}\se_{c}[\semp/x]{\equiv} aay{=}yba} and
\form{\se_{c}^2{\equiv}\se_{c}[\semp/y]{\equiv} xaa{=}bax}.
 It is known that for certain words \form{w_1}, \form{w_2} and a variable \form{z}
the word equation \form{z {\cdot} w_1 {=} w_2 {\cdot} z} is satisfied if
there exist words \form{A}, \form{B}
and a natural number \form{i} such that
 \form{w_1{=}A{\cdot}B}, \form{w_2{=}B{\cdot}A}
and \form{z{=}(A{\cdot}B)^i{\cdot}A}.
Therefore, both \form{\se_{c}^1}
and \form{\se_{c}^2} are
unsatisfiable. The soundness of the cyclic proof implies that
 \form{\se_{c}} is unsatisfiable.
 For this problem, while {\deci} returned
{\unsat}
within 1 second,
 Z3str2 and Z3str3 
 returned {\unknown},
S3P, Norn
and CVC4  were unable to decide
within 180 seconds.


\section{Related Work and Conclusion} \label{sec.related}

Makanin notably provides a mathematical proof for the satisfiability problem
of word equation \cite{Makanin:math:1977}. 
In the sequence of papers,
 Plandowski {\em et.al.} showed that the complexity of
this problem is PSPACE \cite{Plandowski:JACM:2004}.
The proposed procedure {\allSat} is closed to
the (more general) problem in computing
the set of all solutions for a word equation
  \cite{Jaffar:JACM:1990,Plandowski:STOC:2006,Ferte2014,Jez:JACM:2016,Ciobanu:ICALP:2015}.
The algorithm presented
in \cite{Jaffar:JACM:1990} which is based on
Makanin's algorithm does not terminate if the set is infinite.
Moreover,  the length constraints
 derived by \cite{Plandowski:STOC:2006,Jez:JACM:2016}
may not be in a finite form.
\hide{In particular, although Plandowski's work \cite{Plandowski:STOC:2006}
can derive a finite representation of all solutions
of a word equation, the length constraints implied by this representation
 are not always represented with finitely many
 equations in numeric solvable form.}
In comparison, due to the consideration of
cyclic solutions,  {\allSat} terminates even for infinite
sets of all solutions.
The description of the sets of all solutions
as EDTOL languages
was known \cite{Ferte2014,Ciobanu:ICALP:2015}.
For instance, authors in \cite{Ferte2014}
show that the languages of quadratic word equations
can be recognized by some pushdown automaton of level 2.
Although \cite{Jez:JACM:2016} did not aim at giving such
a structural result, it provided {\em recompression} method
which is the foundation for the remarkable procedure in
\cite{Ciobanu:ICALP:2015}
which prove that languages of
solution sets of
 arbitrary word equations
are EDTOL.
In this work, we propose
a decision procedure which is
 based on the description of solution sets
as {\em finite-index} EDTOL languages.
Like \cite{Ferte2014},
we also show that sets of all solutions
of quadratic word equation are EDTOL languages.
In contrast to \cite{Ferte2014},
we give a concrete procedure to construct
such languages for a solvable equation
such that an implementation of the decision procedure
for string constraints
is feasible.
As shown in this work,
 finite-index feature is the key to obtain
a decidability result when handling a theory combining
word equations with length constraints over words.
It is unclear whether the description
derived by the procedure in \cite{Ciobanu:ICALP:2015}
is the language of finite index.
Furthermore, node of the graph derived by \cite{Ciobanu:ICALP:2015}
is an extended equation which is an element in
a free partially commutative monoid rather
than 
 a word equation.

Decision procedures for
quadratic word equations
are presented in \cite{Schulz:1990:MAW,Diekert1999}. 
Moreover, Schulz \cite{Schulz:1990:MAW} also extends Makanin's
algorithm to a theory of word equations and regular memberships.
Recently, \cite{Hooimeijer:PLDI:2009,Hooimeijer:ASE:2010}
presents a decision procedure for subset constraints over regular
expressions.
\cite{Liang:FroCos:2015} presents a decision procedure for
regular memberships and length constraints.
\cite{Ganesh:HVF:2012,Abdulla:CAV:2014} presents a decidable fragment
of {\em acyclic} word equations, regular expressions
and constraints over length functions.
It can be implied that this fragment is subsumed by ours.
\cite{Lin:POPL:2016,Tao:POPL:2018,Lukas:POPL:2018} presents a straight-line
fragment
including word equations and transducer-based functions (e.g., \code{replaceAll})
which is incomparable to our decidable fragments.
%
\hide{In the rest of this section,
we summarize the development of related works on practical string solvers.

{\em Automaton-based Solvers.} Finite automata provides a natural encoding
for string with regular membership constraints.
 Rex \cite{Veanes:ICST:2010} encodes strings as symbolic finite automata (SFA).
Each SFA transition is transformed into SMT constraints.
Java String Analyzer (JSA) \cite{Christensen:SAS:2003} is specialized
for Java string constraints. JSA approximates string constraints 
into multi-level automaton.
 \cite{Hooimeijer:PLDI:2009,Hooimeijer:ASE:2010}
provides a reasoning over string with {\em priori} length bounds.
Recent work in \cite{Aydin:CAV:2015} provides
a length-bound approach for solving string constraints and further
counting the number of solution to such constraints.
Recently, authors in \cite{Abdulla:CAV:2014,Abdulla:CAV:2015}
proposes a DPLL(T)-based approach to
unbounded string constraints with regular expressions and length function.
 \cite{Wang2016} described
a new method based on a scalable logic circuit representation to support
 various string and automata manipulation operations and
 counter-example generation.
The work in \cite{Parosh:PLDI:2017} also uses Parikh's Theorem
to reduce words into letter-counts.
However, this work is not complete
and only aims for an efficient analysis.
In contrast to \cite{Parosh:PLDI:2017},
we base on Parikh's Theorem to construct a
decision procedure for the powerful fragments.

{\em Bit-vector-based Solvers.}
Hampi  \cite{Kiezun:ISSTA:2009}
reduces fixed-sized string constraints to bit-vector problem
and then satisfiability.
 Kazula \cite{Saxena:SP:2010} extends Hampi with
{\em concatenation} operation. It
first solves arithmetical constraints and then
 enumerates possible fixed-length versions
of an input formula using Hampi.
In \cite{Bjorner:TACAS:2009}, strings are represented as arrays.
Discharging string with length constraints are performed through two phases.
First an integer-based over-approximation of the string constraint is solved. After that,
fixed-length string constraints are then decided in a second phase.

{\em Word-based Solvers.}}
Z3str \cite{Zheng:FSE:2013} implements string theory as an extension
of Z3 SMT solver through string plug-in. It supports unbounded string constraints with
a wide range of string operations. Intuitively,
it solves string constraints and generates string lemmas
to control with Z3's congruence closure core.
Z3str2 \cite{Zheng:CAV:2015} improves Z3str by
proposing a detection of those constraints beyond the tractable fragment,
 i.e.  overlapping arrangement,
and pruning the search space for efficiency.
Similar to Z3str, CVC4-based string solver \cite{Liang:CAV:2014}
communicates with CVC4's equality solver to exchange information over string.
S3P \cite{Trinh:CAV:2016} enhances Z3str
to incrementally interchange information between string and arithmetic constraints.
S3P also presented some heuristics
to detect and prune non-minimal subproblems
while searching for a proof.
While the technique in S3P was able to detect non-progressing
scenarios of satisfiable formulas, it would not terminate
for
 unsatisfiable formulas
due to presence of multiple occurrences of each string variable.
Our solver can support well for both classes of queries in case of less than
or equal to two  occurrences of each string variable.

\paragraph{Conclusion}\label{sec.conclude}
We have presented the solver
 {\deci} for the satisfiability
 of string constraints combining
word equations, regular expressions
and length functions.
We have identified two decidable fragments 
including quadratic
word equations. 
Finally, we have implemented and evaluated
{\deci}.
Although our solver is only a prototype, the results are encouraging for their
coverage as well as their performance.
For future work, we plan 
to 
 support
other string operations (e.g., replaceAll).
\hide{
for the full fragment {\strel} as well as
 support
additional practical functions (e.g., replaceAll, charAt).} 

\paragraph{\bf Acknowledgements.}
 Anthony W. Lin for the helpful discussion while the author was visiting Oxford University in December 2017. We are grateful for the constructive feedback from the reviewers of POPL'18 and CAV'18.


\bibliographystyle{abbrv}
\bibliography{all}

  \appendix

\section{String Constraints (Cont)}

\paragraph{Semantics of String Constraint} \label{app.sem}
\begin{figure}[h]
 \begin{center}
  \begin{minipage}{0.65\textwidth}
  \begin{frameit} \savespace  \savespace
\[
\begin{array}{lcl}
\form{\sstack,\istack} \force \pure_1 {\vee} \pure_2 & {\iffs} &
 \form{\sstack,\istack} \force \pure_1 \text{ or }
 \form{\sstack,\istack} \force \pure_2 \\
 \form{\sstack,\istack} \force \pure_1 {\wedge} \pure_2 & {\iffs} &
 \form{\sstack,\istack} \force \pure_1 \text{ and }
 \form{\sstack,\istack} \force \pure_2 \\
\form{\sstack,\istack}  \force \neg \pure_1 & {\iffs} &
 \form{\sstack,\istack} \not{\models} \pure_1 \\
\form{\sstack,\istack} \force s{\in}\regex & {\iffs} & 
\exists w{\in}  \classlang{\regex} \cdot \sstack,\istack \force
 s = w
  \\
 \form{\sstack,\istack} \force s_1{=}s_2 & {\iffs} & 
 {\sstack}(s_1){=} {\sstack}(s_2) \text{ and }
{\istack}(s_1){=} {\istack}(s_2)
  \\
\form{\sstack,\istack} \force s_1{\neq}s_2 & {\iffs} &
\form{\sstack,\istack} \force \neg (s_1{=}s_2)
  \\
 \form{\sstack,\istack} \force {\a_1}{\oslash}{\a_2} & {\iffs} &  {\sstack}(\a_1)~{\oslash}~ {\sstack}(\a_2) \text{, where }
 \oslash \in \{=,\leq\}
  \\
 \end{array}
 \]
 \end{frameit}
 \end{minipage}
\caption{Semantics}\label{prm.sem.fig}
\end{center}  \savespace 
\end{figure}

We show the details of the semantics in Fig. \ref{prm.sem.fig}.
We notice that
an equation of string terms is satisfied
if there exists an assignment that satisfies
word equation over string variables
as well as equation over their lengths.

\subsection{Normalized Formulas} \label{app.norm.form}

We show how to normalize word equations and
regular expressions in a formula.
First, we show how to transform
negation over word equations,
and disjunction of word equations
into
  an equivalent single word equation.
By doing so, it is safe
to consider only single word equation
in the proposed algorithms.
The reader is referred to \cite{khmelevski1976equations,Volker:2002:Book}
for the
correctness of the transformation.
Word disequalities can be eliminated
using the following proposition \cite{Volker:2002:Book}.
\begin{proposition}
A disequality \form{s_1{\neq}s_2} is equivalent with
the following formula:
\[
\bigvee_{a \in \alphabet} (s_1{=}s_2{\cdot}a{\cdot}x \vee
 s_2{=}s_1{\cdot}a{\cdot}x) \vee
 \bigvee_{a,b \in \alphabet, a{\neq}b}(s_1{=}x{\cdot}a{\cdot}y {\wedge}
s_2{=}x{\cdot}b{\cdot}z)
\]
where \form{x}, \form{y} and \form{z} are fresh variables.
\end{proposition}
Intuitively, two string terms \form{s_1} and \form{s_2}
are different
if there exists an interpretation \form{\sstack} and a non-negative number
 \form{i}
such that the letters at position i in \form{\sstack(s_1)} and
\form{\sstack(s_2)} are different.
This elimination is utilized in Norn solver \cite{Abdulla:CAV:2014}.

A disjunction of word equations can be replaced by
  a single word equation as follows.
\begin{proposition}
Let \form{a, b \in \alphabet} be distinct letters and
 \form{a \neq b}.
A disjunction of two word equations is equivalent with a single
word equation in two extra unknowns.
\end{proposition}

Next, we show how to remove the negation and the concatenation
operator
over regular expression. It is easy to show that
the negation of a membership predicate in
a regular expression is equivalent \form{\regex} with
a membership predicate in its complement  \form{\regex^C}.
\begin{lemma}
Let \form{s} be a string term and \form{\regex}
a regular expression. Then,
\form{\neg(s \in \regex) \equiv (s \in \regex^C)}.
\end{lemma}

We note that either
the expression \form{w \in \regex} or the expression \form{\neg(w \in \regex)} where \form{w \in \alphabet^*} is
trivially evaluated and replaced by \form{\true}
or \form{\false}.
Removing the concatenation operator
in the expression \form{s_1 \cdot s_2 \in \regex}
relies on the following function.
Let \form{L} be a regular language
and \form{f(L)} is a set of pairs of DFAs
\form{(D_1, D_2)} which represent for
regular
languages \form{(L_1, L_2)}
such that \form{L = L_1 \cdot L_2}.
To compute the set \form{f(L)} of
a given regular language \form{L}, represent
\form{L} by some fixed automaton \form{\dfa{Q}{\alphabet}{\transition}{s_0}{Q_F}}.
For any state \form{q_i \in Q}, we generate two automata
\form{D_1} form{\dfa{Q}{\alphabet}{\transition}{s_0}{\{s_i\}}} and
\form{D_2} \form{\dfa{Q}{\alphabet}{\transition}{s_i}{Q_F}}, respectively.
 Then, \form{\forall w \in L}
and \form{w=w_1 \cdot w_2}, there exists
a state \form{s_i} to form such two automata which,
in turn, generate two corresponding languages \form{L_1}, \form{L_2}
such that \form{w_1 \in L_1}
and \form{w_2 \in L_2}.
Let \form{DFA2RE} be the function to convert
a DFA to a regular expression. Then, the
 following lemma is straightforward.
\begin{lemma}
Let \form{s_1}, \form{s_2} be string terms and \form{\regex}
a regular expression. Then,
\[
\form{(s_1 \cdot s_2 \in \regex) \equiv \bigvee \{
 s_1 \in DFA2RE(D_1) \wedge s_2 \in DFA2RE(D_2) \mid (D_1,D_2) \in f(\lang(\regex))} \}
\]
\end{lemma}


\hide{ We separate the conjuncts of a formula \form{\pure} into
four parts: \form{\pure{\equiv} s_1{=}s_2 ~{\wedge}~ \astart ~{\wedge}~ \myit{I} ~ {\wedge} ~\subterm } where
(i) \form{s_1{=}s_2} is a word equation,
(ii) \form{\astart} a conjunction of 
regular expressions,
(iii) \form{\myit{I}} is a conjunction of  arithmetical constraints,
 (iv) and finally \form {\subterm} is a conjunction
 of subterm relations  obtained from
unfolding inductive string predicates.
We notice that if it is unambiguous, we sometimes use 
\form{\astart}, \myit{I} and \form{\subterm} as sets
instead of conjunctions. And
while string variables
in \form{\seqs} may be encoded with the inductive predicates,
those in \form{\subterm} are not.
For every string \form{x}, 
its invariant \form{\code{LEN}(x){\geq}0}
 must be implied by \form{\myit{I}}.
Each \form{\subterm} is of the form
either \form{s_1{=}c{\cdot}s_2} or \form{s_1{=}s_2{\cdot}s_3}.
They are deduced during solving a formula
and dedicated for constructing a model to witness {\sat}.}

\section{Correctness of {\allSat} - Proposition \ref{thm.tree}}\label{app.correct.tree}


%

The correctness of procedure {\allSat}
replies on the correctness of \code{\matchse},
\code{\complete} and the soundness
of the {\em cyclic proofs} where all leaf nodes are 
marked as closed.

\paragraph{Procedure \code{\matchse}}
First, we show that
\code{\matchse} produces a equi-satisfiable word
equation.
\begin{lemma}[Matching]\label{lem.sound.match}
Suppose that \form{\se} is a word equation,
and \form{\se'{=}\code{\matchse}(\se)}. Then,
 a) if \form{\se} is satisfiable, so is \form{\se'}.
b) if \form{\se'} is satisfiable, so is \form{\se}.
c) in both cases a) and b),
{\se'} is a suffix of {\se}.
\end{lemma}
\repconf{The proof is presented in App. \ref{proof.lem.sound.match}.}{}

Function \form{\matchse(\se)} also has the following property.
\begin{lemma}\label{lem.match.size}
Let \form{\code{\se'(N') {=}\matchse}(\se(N))}. Then,
  \form{N'{\leq}N}.
\end{lemma}
\repconf{The proof is presented in App. \ref{proof.lem.match.size}.}{}

\paragraph{Procedure {\complete}}
Next, we show that {\complete} produces
an equi-satisfiable set of word equations.
We remark that procedure {\complete} also
produces substitutions labeled along
path traces
which help to construct a model (assignments to string variables)
for satisfiable inputs.

\begin{lemma}[Complete]\label{lem.sound.reduce}
Suppose that \form{\se} is a word equation,
and \form{L} is set of pairs of word equation
and substitutions such that
\form{L=}\complete(\form{\se}). Then,
\begin{itemize}
\item C1) if \form{\se} is satisfiable then
there exists a pair \form{(\se', \sub) \in L}
such that \form{\se'} is satisfiable.
\item C2) if there exists a pair \form{(\se', \sub) \in L}
such that \form{\se'} is satisfiable, then
\form{\se} is satisfiable.
\end{itemize}
\end{lemma}
\repconf{The proof is presented in App. \ref{proof.lem.sound.reduce}.}{}

\paragraph{Cyclic Proofs}
Finally, we consider the case where
the input is unsatisfiable.
Suppose that {\allSat} takes a word equation \form{\se} as input,
and produces a cyclic reduction tree  \form{\utree{n}}
 as output in a finite time.
If all leaf nodes of  \form{\utree{n}}
is unsatisfiable, then following
Lemma \ref{lem.sound.match} and Lemma
 \ref{lem.sound.reduce} we can conclude that \form{\se}
is unsatisfiable.
Now, we study the scenarios
 where some leaf nodes of  \form{\utree{n}}
is unsatisfiable and the remaining leaf nodes are linked back.
We refer to such reduction tree \form{\utree{n}}
as cyclic proofs. In the following,
we show that if {\allSat} can derive
sound cyclic proofs for a word equation \form{\se},
then \form{\se} is unsatisfiable.
 The following formalism is based on the generic
framework S2SAT \cite{Loc:CAV:2016,Le:CAV:2017}.
In contrast to  \cite{Loc:CAV:2016,Le:CAV:2017}, our soundness proof
is based on the fact that
solutions of a word equation must be finite.

\begin{defn}[Pre-proof]
A {\em pre-proof} derived for an equation \form{\se}
is a pair ($\utree{i}$, $\mathcal L$) where
   $\utree{i}$ is an unfolding tree
 whose root labelled by  \form{\se}
and $\mathcal L$ is a back-link function
assigning some leaf nodes \form{\se_{c}} of $\utree{i}$ to interior nodes $\se_{c}= \mathcal L(\se_{b})$
such that there exists some substitution $\theta$ i.e.,
 $\se_{c}=\se_{l}[\theta]$.
\end{defn}
We recap that in the above definition
\form{\se_{b}} is referred as a bud and \form{\se_{c}} is referred as its companion.

A {\em cycle path} in a pre-proof is a sequence of nodes $(\se_{i})_{i{\geq}0}$.

\begin{defn}[Cycle Trace]
 Let $({\se}_{i})_{i{\geq}0}$ be a
cycle path in a pre-proof $\mathcal {PP}$.
 A cycle trace following
$(\se_{i})_{i{\geq}0}$ is a sequence $(\alpha_i)_{i{\geq}0}$ such that, for all
$i{\geq}0$, $\alpha_i$ is a string variable $x$ in the formula
 $\se_{i}$, and either:
\begin{enumerate}
\item $\alpha_{i{+}1}$ is the variable
 $x$
 occurrence in $\se_{i{+}1}$, or
\item  $\alpha_{i{+}1}$ is the subformula \form{c{\cdot}x'} (where $c$ is a letter)
according to \form{x} in \form{\se_{i{+}1}} (i.e., \form{\se_{i{+}1}= \se_i[c{\cdot}x'/x]}) and
 $i$ is a progressing point of the trace, or
\item $\alpha_{i{+}1}$ is the subformula \form{y{\cdot}x'} (where $y$ is a string variable)
according to \form{x} in \form{\se_{i{+}1}} (i.e., \form{\se_{i{+}1}= \se_i[y{\cdot}x'/x]}) and
 $i$ is a progressing point of the trace.
\end{enumerate}
\end{defn}

To ensure that pre-proofs correspond to sound proofs,
a global
{\em soundness condition} must be imposed on such  pre-proofs as follows.
\begin{defn}[Cyclic proof]\label{cyclic}
A pre-proof is a cyclic proof if, for every infinite path 
 $(\form{\se_{i}})_{i{\geq}0}$,
there is a tail of the cycle path $p{=}(\form{\se_{i}})_{i{\geq}n}$
 s.t. there is an
infinitely progressing trace following $p$.
\end{defn}

\begin{lemma}[Soundness]\label{lem.unsat}
If there is a cyclic proof of \form{\se}, \form{\se} is unsatisfiable.
\end{lemma}
\repconf{The proof is presented in App. \ref{proof.lem.unsat}.}{
}

The correctness of Proposition \ref{thm.tree}
 immediately follows the following Lemma \ref{lem.sound.match}, Lemma
 \ref{lem.sound.reduce} and Lemma \ref{lem.unsat}.

\subsection{Proof of Lemma \ref{lem.sound.match}} \label{proof.lem.sound.match}
\begin{proof}
We prove this lemma through the following three cases.
\begin{enumerate}
\item Case 1: \form{\se{\equiv}c \cdot tr_1{=}c \cdot tr_2}
where \form{c} is a letter, then \form{\se'{\equiv}tr_1{=}tr_2}.
The proof is as follows.
\[
\begin{array}{lll}
 & \se \text{ is satisfiable} & \\
 \Leftrightarrow &  \text{there exists
an assignment } \form{\sstack \in {\SStore} } & \\
 &  \text{  such that } \form{\sstack {\force} c \cdot tr_1 {=}c \cdot tr_2} & \\
\Leftrightarrow & \form{\sstack(c \cdot tr_1) = \sstack(c \cdot tr_2)} & \text{meaning of } \sstack \text{ relation} \\
\Leftrightarrow & \sstack(c) \cdot \sstack(tr_1) = \sstack(c) \cdot \sstack(tr_2) & \text{meaning of } \sstack \text{ on concatenation} \\
\Leftrightarrow & \sstack(tr_1) = \sstack(tr_2) & \text{meaning of concatenation}  \\
\Leftrightarrow & \text{there exists
an assignment } \form{\sstack \in {\SStore} } & \\
  & \text{  such that } \sstack {\force} tr_1 {=} tr_2 & \text{meaning of } \sstack\\
\Leftrightarrow &  \form{\se'} \text{ is satisfiable} & 
\end{array}
\]
Furthermore, if \form{S} and \form{S'} are solution words
of \form{\se} and \form{\se'}, respectively, we can imply that
\form{S {=} c \cdot S'}. Hence, \form{\se'} is a suffix of
\form{\se}.
\item Case 2:  \form{\se{\equiv}utr_1{=}utr_2}
where \form{u} is a string variable,
then \form{\se'{\equiv}tr_1{=}tr_2}. We consider two sub-cases.

\noindent Sub-case 2.1: \form{u \in (\FV(tr_1) \cup \FV(tr_2))}.
The proof for this sub-case is similar to the proof in Case 1.

\noindent Sub-case 2.2: \form{u \not\in (\FV(tr_1) \cup \FV(tr_2))}.
The proof for Case a) of this sub-case is similar to the
proof in Case 1. In the following, we show the proof for Case b).
\[
\begin{array}{lll}
 & \se' \text{ is satisfiable} & \\
\Leftrightarrow & \text{there exists
an assignment } \form{\sstack \in {\SStore} } 
\text{  such that } \form{\sstack {\force} tr_1 {=} tr_2} & \\
\Leftrightarrow & \form{\sstack(tr_1) = \sstack(tr_2)} & \text{meaning of } \sstack \\
\end{array}
\]
We create a new assignment \form{\sstack'}
such that i) for all \form{v \in (\FV(tr_1) \cup \FV(tr_2))},
 \form{\sstack'(v)=\sstack(v)}
and ii) \form{\sstack'(u)=\semp}.
From i), we have:
\[
\begin{array}{lll}
& \form{\sstack'(tr_1) = \sstack'(tr_2)} & \\
\Rightarrow & \form{\sstack'(u) \cdot \sstack'(tr_1) = \sstack'(u) \cdot \sstack'(tr_2)} & \text{from ii)} \\
\Leftrightarrow & \form{\sstack'(u \cdot tr_1) = \sstack'(u \cdot tr_2)} & \text{meaning of } \sstack'  \text{ on concatenation}\\
\Leftrightarrow & \text{there exists
an relation } \form{\sstack' \in {\SStore} } & \\
 & \text{  such that } \sstack' {\force} u \cdot tr_1 {=} u \cdot tr_2 & \text{meaning of } \sstack'\\
\Leftrightarrow &  \form{\se} \text{ is satisfiable} & 
\end{array}
\]
Furthermore, if \form{S'} and \form{S} are solution words
of \form{\se'} and \form{\se}, respectively, we can imply that
\form{S {=} S'}. Hence, \form{\se'} is a suffix of
\form{\se}.
\item Case 3:  \form{\se{\equiv}u_1tr_1{=}u_2tr_2} and \form{\se'{\equiv}\se}
where \form{u_1}, \form{u_2} are either letters or string variables
and \form{u_1}, \form{u_2} are different.
The proof for this case is straightforward.
\end{enumerate}
\qed
\end{proof}

\subsection{Proof of Lemma \ref{lem.match.size}} \label{proof.lem.match.size}
\begin{proof}
Based on the definition of {\matchse},
 we consider two following cases.
\begin{itemize}
\item M1) \form{\se \equiv u tr_1 = u tr_2} (where \form{u} is a string variable
or a letter), then
 \form{\se' \equiv tr_1 = tr_2}. It is easy to show that
\form{N'=N-2}.
\item M2) \form{\se \equiv u_1 tr_1 = u_2 tr_2}
(where \form{u_1} and \form{u_2} are string variables
or letters), then
 \form{\se' \equiv \se}. Hence, \form{N'=N}.
\end{itemize}
\hfill \qed
\end{proof}

\subsection{Proof of Lemma \ref{lem.sound.reduce}} \label{proof.lem.sound.reduce}
\begin{proof}
Based on the definition of {\complete} procedure,
we prove this lemma (both C1) and C2)) through following two cases.
\begin{enumerate}
\item \form{\se{\equiv}x_1tr_1=c_2tr_2}, then \form{L_i=\{
(\form{\se_{i_1}},\form{\sub_1}); (\form{\se_{i_2}}, \form{\sub_2})\}} where  \form{\rho_1{=}[\semp/x_1]},
 \form{\se_{i_1}{\equiv}(x_1tr_1{=}c_2tr_2)\rho_1},
\form{\rho_2{=}[c_2x_1'/x_1]} ( \form{x_1'} is a fresh variable)
and \form{\se_{i_2}{\equiv}(x_1tr_1{=}c_2tr_2)\rho_2}.

We start with \form{\se} is satisfiable.
\[
\begin{array}{lll}
 & \se \text{ is satisfiable} & \\
\Leftrightarrow & \text{there exists
an assignment } \form{\sstack \in {\SStore} } & \hfill \\
 & \text{  such that } \form{\sstack {\force} x_1 \cdot tr_1 {=}c_2 \cdot tr_2} & \\
\Leftrightarrow & \form{\sstack(x_1 \cdot tr_1) = \sstack(c_2 \cdot tr_2)} & \hfill \text{meaning of } \sstack \text{ relation} \\
\Leftrightarrow & \sstack(x_1) \cdot \sstack(tr_1) = \sstack(c_2) \cdot \sstack(tr_2) & \hfill \text{meaning of } \sstack \text{ on concatenation} \\
\Leftrightarrow & \sstack(x_1) \cdot \sstack(tr_1) = c_2 \cdot \sstack(tr_2) & \hfill \text{meaning of } \sstack \text{ on letters} ~~~~\textcolor{blue}{(a_1)}\\
\end{array}
\]
Now, we do case split on \form{\sstack(x_1)}:
 \form{\sstack(x_1){=}\semp} or \form{\sstack(x_1){=}w_1} and
\form{w_1{\not=}\semp}.

\noindent \underline{Sub-case 1.1}: \form{\sstack(x_1){=}\semp}. We create
new assignment  \form{\sstack'} such that:
\form{\sstack'(x_1)} is not defined and
\form{\sstack'(v) {=} \sstack(v)} \form{\forall v \in (\FV(tr_1) \cup \FV(tr_1) \setminus \{ x_1 \})}.
It is easy to show that
\form{\sstack(tr_1) \equiv  \sstack'(tr_1\sub_1)}
and
\form{\sstack(tr_2) \equiv  \sstack'(tr_2\sub_1)} \textcolor{blue}{($a_2$)}.
From ($a_1$) and ($a_2$), we obtain:
\[
\begin{array}{lll}
 &  \form{\sstack'(tr_1\sub_1) = c_2 \cdot \sstack'(tr_2\sub_1)} & \\
\Leftrightarrow & \form{\sstack'(tr_1\sub_1) =  \sstack'(c_2 \cdot tr_2\sub_1)} & \text{meaning of } \sstack' \text{ on letters}  \\
\Leftrightarrow & \form{\sstack'(tr_1\sub_1) =  \sstack'((c_2 \cdot tr_2)\sub_1)} & \text{substitution does not affect constants}  \\
\Leftrightarrow & \form{\sstack'((x_1 \cdot tr_1)\sub_1) =  \sstack'((c_2 \cdot tr_2)\sub_1)} &   \\
\Leftrightarrow & \form{\sstack'((x_1 \cdot tr_1)\sub_1) =  \sstack'((c_2 \cdot tr_2)\sub_1)} &  \\
\Leftrightarrow & \form{\sstack' {\force}  (x_1 \cdot tr_1)\sub_1 {=} (c_2 \cdot tr_2)\sub_1} & \\
\Leftrightarrow & \form{\sstack' {\force}  (x_1 \cdot tr_1 {=} c_2 \cdot tr_2)\sub_1} & \\
 & \se_{i_1} \text{ is satisfiable} & \\
\end{array}
\]
We can conclude that there exists 
\form{(\form{\se_{i_1}},\form{\sub_1}) \in L_i}
such that \form{\se_{i_1}} is satisfiable.

\noindent \underline{Sub-case 1.2}: \form{\sstack(x_1){=}w} and
\form{w{\not=}\semp}. Substituting into ($a_1$) to obtain:
\form{w_1 \cdot \sstack(tr_1) = c_2 \cdot \sstack(tr_2)}.
This implies that \form{w_1} must start with letter \form{c_2}.
As so, we assume that \form{w_1{=}c_2{\cdot}w_1'}. Hence,
\form{c_2 \cdot w_1' \cdot \sstack(tr_1) = c_2 \cdot \sstack(tr_2)}
\textcolor{blue}{($a_3$)}.

 We create
new assignment  \form{\sstack'} such that:
\form{\sstack'(x_1)} is not defined,
\form{\sstack'(v) {=} \sstack(v)} \form{\forall v \in (\FV(tr_1) \cup \FV(tr_1) \setminus \{ x_1 \})}, and 
\form{\sstack'(x_1') {=} w_1'}.
It is easy to show that
\form{\sstack(tr_1) \equiv  \sstack'(tr_1\sub_2)}
and
\form{\sstack(tr_2) \equiv  \sstack'(tr_2\sub_2)} \textcolor{blue}{($a_4$)}.
From ($a_3$) and ($a_4$), we have:
\[
\begin{array}{lll}
 & c_2 \cdot w_1' \cdot \sstack'(tr_1\sub_2) = c_2 \cdot \sstack'(tr_2\sub_2) & \\
\Leftrightarrow & c_2 \cdot \sstack'(x_1') \cdot \sstack'(tr_1\sub_2) = c_2 \cdot \sstack'(tr_2\sub_2) & \\
\Leftrightarrow &  \sstack'( c_2 \cdot x_1' \cdot tr_1\sub_2) = \sstack'(c_2 \cdot  tr_2\sub_2) & \\
\Leftrightarrow &  \sstack'( (c_2 \cdot x_1' \cdot tr_1)\sub_2) = \sstack'((c_2 \cdot  tr_2)\sub_2) & \text{substitution does not affect constants and } x_1'\\
\Leftrightarrow & \form{\sstack' {\force}  (c_2 \cdot x_1' \cdot tr_1)\sub_2 = (c_2 \cdot  tr_2)\sub_2 } & \\
 & \se_{i_2} \text{ is satisfiable} & \\
\end{array}
\]
We can conclude that there exists 
\form{(\form{\se_{i_2}},\form{\sub_2}) \in L_i}
such that \form{\se_{i_2}} is satisfiable.





\item \form{\se{\equiv}x_1tr_1=x_2tr_2}, then \form{L_i=\{ (\form{\se_{i_1}}, \form{\sub_1});
  (\form{\se_{i_2}}, \form{\sub_2});  (\form{\se_{i_3}}, \form{\sub_3});
  (\form{\se_{i_4}}, \form{\sub_4})\}}.
The proof is
similar to the case above.

\end{enumerate}
\qed
\end{proof}

\subsection{Proof of Lemma \ref{lem.unsat}} \label{proof.lem.unsat}
\begin{proof}
We prove this lemma by contradiction.
Assume there is a cyclic proof $\mathcal {PP}$ of \form{\se}
and \form{\se} is satisfiable.
As the lengths of the solution words are finite,
the lengths of paths starting from the root
to satisfiable leaf nodes must be finite.

By following Lemma \ref{lem.sound.match} and Lemma
 \ref{lem.sound.reduce},
we would be able to
construct an infinite path
 $(\form{\se_{i}})_{i{\geq}0}$ in  $\mathcal PP$
such that \form{\se_{0} \equiv \se}
and \form{\se_{i}} is satisfiable for all \form{i\geq0}.
 Since $\mathcal {PP}$ is a cyclic proof,
there exists an \form{n{\geq}0}
and a tail of the path, $p{=}(\form{\se_{i}})_{i{\geq}n}$, 
such that there is an
infinitely progressing trace following $p$ (Definition \ref{cyclic}).
This contradicts the fact of finite
 path for solutions of a word equation above.
 \qed.
\end{proof}

\section{Correctness of \form{\widentree{\utree{n}, \astart}} - Proposition \ref{lemma.edtl.re}}\label{app.lemma.edtl.re}
We show correctness of extended tree
representing for all solutions
of the conjunction of a word equation, say \form{\se},
 and regular expressions, say \form{\astart}.
The computation of \form{m} and \form{M}
 is based on the proof
presented by Schulz
 in \cite{Schulz:1990:MAW} to find the {\em minimal}
solutions of the constraint \form{\se \wedge \astart}.
First, we transform the constraint on regular expressions
into constraints over DFA.
It is well known that there exists a DFA that
accepts the same language with
a conjunction of regular expressions \cite{Hopcroft:2006:IAT}.


Suppose that we are given a DFA:
 \form{A=\dfa{Q}{\alphabet}{\transition}{q_o}{Q_F}}.
For any pair \form{(q_i,q_j)} in \form{Q},
\form{\classlang{A^{q_i}_{q_j}}} denotes the language
which is accepted by the automaton:
\form{A^{q_i}_{q_j}=\dfa{Q}{\alphabet}{\transition}{q_i}{\{q_j\}}}.
Further, we define \form{\classlang{A,\emptyset}=\alphabet^*}
and \form{\classlang{A,\Gamma}=\bigcap_{(q_i,q_j) \in Q \times Q}\classlang{A^{q_i}_{q_j}}} where \form{\Gamma \neq \emptyset}
and \form{\Gamma \subseteq Q \times Q}.
 An \form{A-constraint} is
 a finite set \form{\Gamma} of pairs
\form{(q_i,q_j) \in  Q \times Q}.
Given a constraint \form{\pure \equiv \se \wedge X_1 {\in} \regex_1 {\wedge} ...{\wedge}X_n {\in} \regex_n } over \form{\alphabet}
where \form{X_1,..,X_n \in \FV(\se)},
It is obvious that
we can find a DFA \form{A} and \form{A-constraints}
\form{\Gamma_1, ...,\Gamma_n} such that
a word \form{w_i \in \alphabet^*} is an assignment for \form{X_i}
of a solution of
\form{\pure} if \form{X_i \in \classlang{A,\Gamma_i}}
(\form{1\leq i \leq n}).
Let \form{m} be the number states of \form{A}
and \form{M=m!}.
\begin{defn}
The natural number \form{t} and \form{t'}
are called \form{A}-equivalent (we write \form{t \equiv_{A} t'})
 if the following two conditions hold:
\begin{enumerate}
\item \form{t = t' \mod M},
\item \form{t > m} if and only if \form{t' > m}.
\end{enumerate}
\end{defn}
\begin{lemma}[\cite{Schulz:1990:MAW}]\label{lemma.Schulz}
Let \form{v\equiv u_1w^{t_1} u_2w^{t_2}... u_kw^{t_k}u_{k{+}1}}
and \form{v'\equiv u_1w^{t'_1} u_2w^{t'_2}... u_kw^{t'_k}u_{k{+}1}} (\form{k\geq1}) be two words over the alphabet \form{\alphabet}
and \form{\Gamma} is any \form{A}-constraint.
If \form{t_i \equiv_{A} t_i'} (\form{1 \leq i \leq k}),
then \form{v \in  \classlang{A,\Gamma}}
if and only if \form{v' \in  \classlang{A,\Gamma}}.
\end{lemma}

\subsection{Proof of Proposition \ref{lemma.edtl.re}}
\label{proof.lemma.edtl.re}
\begin{proof}
An assignment for a variable of a solution
obtained from labels including
any cycle in the resulting tree 
is a word \form{v\equiv u_1w^{t_1} u_2w^{t_2}... u_kw^{t_k}u_{k{+}1}}
where \form{t_i > m} and \form{t_i \mod M = 0}.
Following Lemma \ref{lemma.Schulz},
a word \form{v} obtained  from labels \form{\sub} including
 cycles in the resulting tree
is of a solution if and only if
the word \form{v} obtained  from \form{\sub \setminus \sub_c}
(where \form{\sub_c} are labels obtained from cycles)
is of a solution.
\qed
\end{proof}

\section{\form{\edtl} Languages}\label{sec.dec.cfg.parikh}

 \subsection{L Systems and Finite-Index \form{\edtl} Language}\label{sec.edt0l}
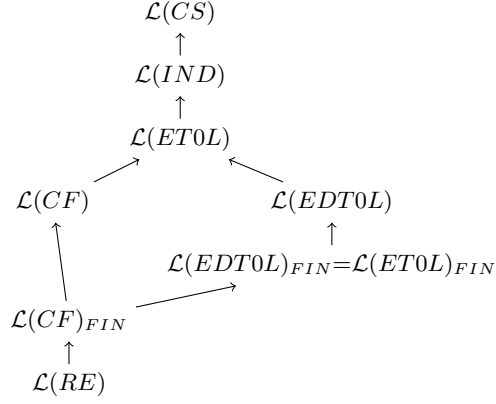
\begin{figure}{t}
\begin{center}
\begin{tikzpicture}[node distance=18mm,level 1/.style={sibling distance=22mm},
      level 2/.style={sibling distance=25mm},
                        level distance=22pt, draw]
  \tikzstyle{every state}=[draw,text=black]

\node (A)                    {\textcolor{black}{$\classlang{CS}$}};
  \node         (B) [below =3mm of A] {$\classlang{IND}$};
  \node         (C) [below =3mm of B] {$\classlang{\etl}$};
  \node         (D) [below left=3mm and 3mm of C] {$\cfls$};
  \node         (E) [below right=3mm and 3mm of C] {$\edtls$};
  \node         (F) [below =3mm of E] {$\edtlfs {=} \etlfs$};
  \node         (G) [below left=2mm and 3mm of F] {$\cflfs$};
  \node         (H) [below =3mm  of G] {$\classlang{RE}$};

  \path (B) edge[->]     node [pos=0.3,left] {} (A)
        (C) edge[->]     node [pos=0.3,left] {} (B)
        (D) edge [->]    node [pos=0.3,left] {} (C)
        (E)   edge [->]    node [pos=0.3,right] {} (C)
        (F)   edge [->]    node [pos=0.3,right] {} (E)
        (G)   edge [->]    node [pos=0.3,right] {} (F)
        (G)   edge [->]    node [pos=0.3,right] {} (D)
        (H)   edge [->]    node [pos=0.3,right] {} (G)
;
\end{tikzpicture}
\end{center}\savespace
\caption{
  A solid line
 denotes strict inclusion in the direction indicated.}\label{lang.incl}
 \end{figure}
\form{\edtl} language is a subclass of indexed languages (denoted as \form{IND}) in the sense
of Aho \cite{Aho:JACM:1968}.
In \cite{Aho:JACM:1968}, Aho
shows that the class of indexed languages
 includes all context-free languages
and some context-sensitive languages, but yet is a proper
subset of the class of context-sensitive languages.
Fig. \ref{lang.incl} shows the containment relationships among classes of indexed
languages where 
$\classlang{CS}$
denotes the class of all context-sensitive languages,
and $\classlang{RE}$ the class of all regular languages.
In the name  \form{\edtl}, each capital letter has
 a standard meaning
in connection with $L$ systems \cite{L:Book}.
Thus, \form{L} refers to parallel rewriting. The character $0$
means that information between individual letters is zero-sided,
and $D$ (deterministic) that, for each configuration
(a letter in a context), there is only one rule.
The letter $T$ (tables) means that the rules are divided into
subsets. In each derivation step, rules from the same subset
have to be used; and $D$ in this context means that each
subset is deterministic.
Finally, $E$ (extended) means $L$ is intersected with
 \form{\alphabet^*}.

\subsection{Proof of Corollary \ref{thm.edtf}}\label{proof.thm.edtf}
\begin{proof}
Given a trimmed reduction tree  \form{\utree{n}}, function \form{\extractedtl}
 constructs for it
a {\etl} grammar \form{G=\lgrammar{\var}{\alphabet}{\productions}{S}} as follows.
\form{\alphabet} is the alphabet.
\form{S} is a fresh variable which does not appear in the tree.
\form{\var} is the union of the set of all variables
appearing in the tree and the set \form{\{S\}}.
For each path \form{(\tnode_r,\tnode_{l_i})} in the trimmed
 \form{\utree{n}} where
\form{\tnode_r} is the root
and \form{\tnode_{l_i}} is either
a satisfiable leaf node or a bud of a cycle, we create
a new table \form{\production_i} as:
\[
\production_i = \{S \rightarrow  s_{l}\} \cup \bigcup \{X \rightarrow s \mid (\ses, [s/X], \ses')  \text{ in }  (\tnode_r,\tnode_{l_i}) \}
\]
Assume that we create \form{m} such tables: \form{\production_1},...,\form{\production_m}. Then, \form{\productions = \{\production_1,...,\production_m\}}.

Moreover, as each table \form{\production_i} for all \form{i \in \{1,...,m\}} corresponds to
 a path of the tree, the rule in \form{\production_i} is deterministic.
Hence \form{G} is a \form{\edtl} system.

Finally, let \form{k} be the maximum of the lengths
of all nodes in the trimmed \form{\utree{n}}.
Then, for every node \form{\se_i(N_i)} in the trimmed \form{\utree{n}},
\form{N_i \leq k}. As so, the number of variables appearing
in every node in the tree is less than or equal to \form{k}.
Hence the language generated by \form{G}
is finite index.
\qed
\end{proof}


\section{Proof of Proposition \ref{lemma.allsat.quad}}

\begin{lemma}\label{lem.match.quad}
Let \form{\se} be a quadratic word equation
and \form{\code{\se' {=}\matchse}(\se)}. Then,
  \form{\se'} is also a quadratic word equation.
\end{lemma}
\repconf{The proof is presented in App. \ref{proof.lem.match.quad}.}{}

\begin{lemma}\label{lem.reduce.size}
Let \form{\se(N)} be a quadratic word
equation and \form{\code{L{=}\transform}(\se)}. Then,
for each \form{(\se'(N'), \sub) \in L}, the following two properties are true:
\begin{itemize}
\item P1. \form{\se'} is unsatisfiable or a quadratic equation.
\item P2.  \form{N'{\leq}N}.
\end{itemize}
\end{lemma}
\repconf{The proof is presented in App. \ref{proof.lem.reduce.size}.}{}

The following lemma immediately follows
Lemma \ref{lem.match.size} and Lemma \ref{lem.reduce.size}.
\begin{lemma}\label{lem.tree.size}
Let \se(N) be a quadratic word equation. And
{\allSat} generates for it a reduction
tree
\utree{n} in finite time. For every node
\se'(N') in \utree{n}, \form{\se'} is unsatisfiable or a quadratic
word equation and \form{N'\leq N}.
\end{lemma}

\subsection{Proof of Lemma \ref{lem.match.quad}}
\label{proof.lem.match.quad}
\begin{proof}
Based on the definition of {\matchse},
 we consider two following cases.
\begin{itemize}
\item \form{\se \equiv u tr_1 = u tr_2} (where \form{u} is a string variable
or a letter), then
 \form{\se' \equiv tr_1 = tr_2}. As \form{\se} is a quadratic word
equation, every variable in \form{\FV(tr_1) \cup \FV(tr_2)} occurs
at most twice in  \form{tr_1 = tr_2}. Hence, \form{\se'} is a quadratic
word equation.
\item M2) \form{\se \equiv u_1 tr_1 = u_2 tr_2}
(where \form{u_1} and \form{u_2} are string variables
or letters), then
 \form{\se' \equiv \se}. Trivially.
\end{itemize}
\qed
\end{proof}

\subsection{Proof of Lemma \ref{lem.reduce.size}}
\label{proof.lem.reduce.size}
\begin{proof}
We consider two following cases
based on the definition of procedure {\complete}.
\begin{enumerate}
\item \form{\se(N){\equiv}x_1tr_1=c_2tr_2}, then \form{L_i=\complete(\se)\equiv \{
(\form{\se_{i_1}(N_1)},\form{\sub_1}); (\form{\se_{i_2}(N_2)}, \form{\sub_2})\}} where
 \form{\se_{i_1}{\equiv}(x_1tr_1{=}c_2tr_2)\rho_1},
\form{\rho_1{=}[\semp/x_1]},
 \form{\se_{i_2}{\equiv}(x_1tr_1{=}c_2tr_2)\rho_2}
 (\form{x_1'} is a fresh variable)
and \form{\rho_2{=}[c_2x_1'/x_1]}.
\begin{enumerate}
\item \form{\se_{i_1}{\equiv}(x_1tr_1{=}c_2tr_2)[\semp/x_1]
{\equiv}(tr_1[\semp/x_1]{=}c_2tr_2[\semp/x_1])}. Let \form{\se_1(N_1'){\equiv}\matchse(\se_{i_1})}.
\begin{itemize}
\item P1. As every variable in \form{\FV(\se)\setminus\{x_1\}} occurs at most twice,
 \form{\se_{i_1}} is a quadratic word equation. And following Lemma \ref{lem.match.quad},
\form{\se_1} is also a quadratic word equation.
\item P2. As in \form{\se_{i_1}} \form{x} is substituted by \form{\semp},
\form{N_1\leq N-1}. And following Lemma \ref{lem.match.size},  \form{N_1'\leq N_1}.
In consequence, \form{N_1'\leq N-1}.
\end{itemize}
\item \form{\se_{i_2}{\equiv} (x_1tr_1{=}c_2tr_2)[c_2x_1'/x_1]}.
Let \form{\se_2(N_2'){\equiv}\matchse(\se_{i_2})}.
\begin{itemize}
\item P1. As \form{\se} is a quadratic equation, \form{x_1}
as well as every variable in \form{\FV(tr_1=tr_2) \setminus\{x_1\}}
occurs  at most twice in \form{\se}. Hence, \form{x_1'}
as well as  \form{\FV(tr_1=tr_2) \setminus\{x_1\}}
occurs  at most twice in \form{\se_{i_2}}. In consequence,
 \form{\se_{i_2}} is a quadratic equation.
And following Lemma \ref{lem.match.quad},
\form{\se_2} is also a quadratic word equation.

\item P2. As \form{\se} is a quadratic equation, \form{x_1}
occurs  at most twice in \form{\se}. Hence,
\form{N_2 \leq N+2}.
 Then,
\[
 \form{\se_2(N_2'){\equiv}
\matchse(\textcolor{blue}{c_2}x_1'tr_1{=}\textcolor{blue}{c_2}tr_2)[c_2x_1'/x_1])
{\equiv}\matchse(x_1'tr_1{=}tr_2)[c_2x_1'/x_1])}
\]
Thus, \form{N_2'\leq N_2-2}. In consequence, \form{N_2'\leq N}.
\end{itemize}
\end{enumerate}
\item \form{\se{\equiv}x_1tr_1=x_2tr_2}, then \form{L_i= \complete(\se)\equiv\{ (\form{\se_{i_1}}, \form{\sub_1});
  (\form{\se_{i_2}}, \form{\sub_2});  (\form{\se_{i_3}}, \form{\sub_3});
  (\form{\se_{i_4}}, \form{\sub_4})\}}.
The proof is
similar to the case above.
\end{enumerate}
\qed
\end{proof}

\subsection{Proof of Proposition \ref{lemma.allsat.quad}}
\label{proof.lemma.allsat.quad}
\begin{proof}
We show the following property, called PathLength property:
Let \form{\se(N)} 
 be a quadratic word equation, then
the length of every path in a reduction
tree whose root is \form{\se(N)}
is $\mathcal{O}(N^2(N!))$.

We prove the PathLength property by structural induction on \form{N}.

\noindent {\bf Base Case}: N=1. Trivially.

\noindent {\bf Induction Case}. Assume that PathLength property
holds 
for all \form{N \leq k}
We now prove that 
PathLength property
holds for \form{N =k+1}.
Consider a path with the sequence of nodes:
\form{\se(k+1)}, \form{\se_1(N_1)}, ...,\form{\se_l(N_l)}
where \form{\se_l(N_l)} is a leaf node.
According to Lemma \ref{lem.tree.size}, we consider two following subcases.
\begin{enumerate}
\item All these nodes have the same length with \form{\se}: \form{k+1 = N_1=...=N_M}.
There are $\mathcal{O}((k+1)!)$ possibilities to arrange a sequence of \form{N} symbols of the respective either string variables or characters.
And there is $\mathcal{O}(k)$ possibilities to put symbol \form{=}
into one sequence of \form{k+1} symbols above.
Thus, we have $\mathcal{O}(k(k+1)!)$ possibilities. In other words,
the length of this path is $\mathcal{O}((k+1)(k+1)!)$.
As \form{(k+1)^2=  k^2 + 2k + 1} and \form{k+1 < k^2 + 2k + 1}
for all \form{k{\geq}1},
PathLength property holds.
\item There exists a node \form{\se_M(N_M)} (\form{1\leq M < l}) such that
(i) all nodes \form{\se_j(N_j)} where \form{j \in \{1....M\}} have the same length with \form{\se}: \form{k+1 = N_1=...=N_M} and (ii) node \form{\se_{M+1}(N_{M+1})}
has the smaller length than \form{\se}.
By induction, the length of the path (\form{\se_{M+1}}, \form{\se_l}) is $\mathcal{O}(k^2(k)!)$.
Similar to the previous sub-case,
the length of the path (\form{\se}, \form{\se_M})
 is $\mathcal{O}(k(k+1)!)$.
Hence, the length of the path (\form{\se}, \form{\se_l}) is:
$\mathcal{O}(k^2(k)!) + \mathcal{O}(k(k+1)!) = \mathcal{O}((k^2/(k+1) + k) (k+1)!)$
As \form{(k+1)^2=  k^2 + 2k + 1} and \form{k^2/(k+1) + k < k^2 + 2k + 1} for all \form{k{\geq}1},
PathLength property holds.
\end{enumerate}
%
\qed
\end{proof}

\section{Correctness of \form{\strflatsem} Decidable Fragment}



\subsection{Straight-Line Formula}\label{proof.prop.acyclic}
\begin{lemma}\label{prop.acyclic}
Let \se(N) be an acyclic word equation.
Then, {\allSat} takes \se(N) as input
and runs in $\mathcal{O}(N)$ time
 to construct
a cyclic reduction tree
\utree{n}. Furthermore, \utree{n} does not
contain any cycle.
\end{lemma}
\begin{proof}
Let \form{\se(N)} be a quadratic word
equation and \form{\code{L{=}\complete}(\se)}.
It is easy to show that
for all \form{\se'(N') \in L}, \form{N'\leq N{-}1}.
Let \form{\se(N)} be a quadratic word
equation and \form{\code{L{=}\transform}(\se)}. Then,
together with Lemma \ref{lem.match.size},
we have: for all \form{\se'(N') \in L}, \form{N'\leq N{-}1}.
As each step, {\allSat} reduces the length
 of word equation by at least one,
the length of each path in the reduction tree is at most $\mathcal{O}(N)$ .

Furthermore, as the lengths of children are always less than
their parents, function {\code{\lb}} never successfully links
a child back to a interior node. Thus, the tree has no cycle.
\qed
\end{proof}

\subsection{Base Case of Proposition \ref{lemma.cfg.phased.twisted.prop}}
\label{proof.lemma.cfg.twisted.prop}
\begin{lemma}\label{lemma.cfg.twisted.prop}
Suppose that
 \form{\ses} is a solvable regular word equation. Then,
 {\allSat} takes \form{\se(N)} as input,
and produces a cyclic reduction tree  \form{\utree{n}}
  in a finite time. Furthermore, for
any cycle \form{\ctree{\ses_c}{\ses_b}{\sub_{cyc}}}
 of \form{\utree{n}}, both three following properties hold:
\begin{itemize}
\item The labels along the path (\form{\ses_c}, \form{\ses_b})
 (assume that this
 path has \form{k} edges) is of the form:
\form{[c_1X_1/X]}, \form{[c_2X_2/X_1]},..., \form{[c_{k+1}X_{k+1}/X_{k}]}
where \form{X}, \form{X_i} (\form{i \in \{1,...,k+1 \}}) are string
variables
and \form{c_i} (\form{i \in \{1,...,k+1 \}}) is 
a letter.
\item The substitution \form{\sub_{cyc}= [X/X_{k+1}]}.
\item The length of each path in the tree is $\mathcal{O}(N)$.
\end {itemize}
\end{lemma}
\begin{proof}
We prove this lemma by structural induction
on \form{n}, the number of duals in the input.
We note that in our following proof,
\form{s_1}, \form{s_2}, \form{s_3}
and \form{s_4} are string terms
and
 \form{X}, \form{Y} are string variables.

\noindent {\bf Case n= 0.} The truth of this case
is shown by
Lemma \ref{prop.acyclic}.

\noindent {\bf Case n= 1.} Wlog, assume that
\form{\se\equiv Xs_1= s_2Xs_3} and \form{X} does not
occur in \form{s_1}, \form{s_2}, \form{s_3}. We consider two following cases.
\begin{enumerate}
\item \form{s_2{\equiv}as_2^1 } where a is a letter.
Then, \form{\se\equiv Xs_1=as_2^1Xs_3} where  all variables in \form{s_1},
\form{s_2^1} and \form{s_3} occur at most once.
Following the definition of function \code{\complete},
\form{\se} has two children \form{\se_1} and \form{\se_2} as follows.
  \begin{enumerate}
    \item \form{\se_1\equiv \se[\semp/X]}. As all variables in 
       \form{\se_1} occur at most one, there is no cycle in the subtree
       with the root is \form{\se_1} (Lemma \ref{prop.acyclic}).
     \item \form{\se_2\equiv \se[aX_1/X]\equiv X_1s_1=s_2^1aX_1s_3}. 
  \end{enumerate}
\item \form{s_2{\equiv}Ys_2^1 } where \form{Y} is a string variable.
 \form{\se\equiv Xs_1=Ys_2^1Xs_3} where \form{Y} and all variables in \form{s_1},
\form{s_2^1} and \form{s_3} occur at most once.
Following the definition of function \code{\complete},
\form{\se} has four children \form{\se_1}, \form{\se_2},
\form{\se_3} and \form{\se_4} as follows.
  \begin{enumerate}
    \item \form{\se_1\equiv \se[\semp/X]}. Similar to Case 1a).
    \item \form{\se_2\equiv \se[YX_1/X]\equiv X_1s_1=s_2^1YX_1s_3}.
    \item \form{\se_3\equiv \se[\semp/Y]}. By induction.
    \item \form{\se_4\equiv \se[XY_1/Y]\equiv s_1=Y_1s_2^1Xs_3}. Now, all variables
    in \form{\se_4} occurs at most once.  Similar to Case 1a).
   \end{enumerate}
\end{enumerate}
If Case 1b) or Case 2b) are kept applying,
after \form{k=\sleng{s_2}} times the node generated
is \form{\se_{k} \equiv X_{k+1}s_1=s_2X_{k+1}s_3}.
Then, function {\lb} links \form{\se_{k}} back to \form{\se} to form
a cyclic proof. It is easy to check that the Lemma holds for this scenario.

\noindent {\bf Case n=2.}
 \form{\se\equiv Xs_1Y=Ys_2X} where
\form{X} and \form{Y} do not
occur in \form{s_1}, \form{s_2}
and
each variable in \form{s_1}, \form{s_2} occurs at most once.
  We consider
following two cases.
Following the definition of function \code{\complete},
\form{\se} has four children \form{\se_1}, \form{\se_2},
\form{\se_3} and \form{\se_4} as follows.
\begin{enumerate}
  \item \form{\se_1\equiv \se[\semp/X]}. By Case n=1.
  \item \form{\se_2\equiv \se[YX_1/X]\equiv X_1s_1Y=s_2YX_1}.
    Consider two following subcases.
    \begin{enumerate}
      \item \form{s_2\equiv as_2^1} and \form{\se_2\equiv  X_1s_1Y=as_2^1YX_1} where \form{a} is a letter. Then, following the definition of function \code{\complete},
        \form{\se_2} has two children:
          \begin{enumerate}
           \item \form{\se_{21}\equiv(X_1s_1Y=as_2^1YX_1)[\semp/X_1]\equiv s_1Y=as_2^1Y}. Case n=1.
           \item \form{\se_{22}\equiv(X_1s_1Y=as_2^1YX_1)[aX_2/X_1]\equiv X_2s_1Y=s_2^1YaX_2}.
          \end{enumerate}
      \item \form{s_2\equiv Zs_2^1}
         and \form{\se_2\equiv  X_1s_1Y=Zs_2^1YX_1}
         where \form{Z} is a string variable. Then, following the definition of function \code{\complete},
        \form{\se_2} has four children:
           \begin{enumerate}
           \item \form{\se_{21}\equiv(X_1s_1Y=Zs_2^1YX_1)[\semp/X_1]\equiv s_1Y=Zs_2^1Y}. Case n=1.
           \item \form{\se_{22}\equiv(X_1s_1Y=Zs_2^1YX_1)[ZX_2/X_1]\equiv X_2s_1Y=s_2^1YZX_2}.
           \item \form{\se_{23}\equiv(X_1s_1Y=Zs_2^1YX_1)[\semp/Z]\equiv X_1s_1Y=s_2^1YX_1}.
           \item \form{\se_{24}\equiv(X_1s_1Y=Zs_2^1YX_1)[X_1Z_1/Z]\equiv s_1Y= Z_1s_2^1YX_1}. Case n=1.
          \end{enumerate}
     \end{enumerate}
  \item \form{\se_3\equiv \se[\semp/Y]}. By Case n=1.
  \item \form{\se_4\equiv \se[XY_1/Y]\equiv s_1XY_1  =Y_1s_2X}. Similar to Case 2.
\end{enumerate}
  If Case 2.b.ii) or Case 2.b.iv) are kept applying,
after \form{k=\sleng{s_2}} times the node generated
is \form{\se_{k} \equiv X_{k+1}s_1Y=Ys_2X_{k+1}}.
Then, function {\lb} links \form{\se_{k}} back to \form{\se} to form
a cyclic proof. It is easy to check that the Lemma holds for this scenario.
It is similar to Case 4. We remark that Case 2 and Case 4 are never
applied in an interleaving sequence.

\hfill \qed
\end{proof}

\subsection{Proof of Proposition \ref{lemma.cfg.phased.twisted.prop}}
\begin{proof}
Wlog, we assume that  \form{\ses} is
 a phased-regular word equation. 
 We prove this Theorem
by structural induction on \form{n}
where the  Lemma 
 \ref{lemma.cfg.twisted.prop}
is used for
the base case.

\hfill \qed
\end{proof}

\end{document}